\documentclass[12pt]{article}
\usepackage{amssymb,amsmath,amsthm,amsfonts}
\usepackage{setspace}
\usepackage{graphicx}
\usepackage{enumerate}
\usepackage{bm}
\usepackage{bbm} 
\usepackage{dsfont}
\usepackage[authoryear,longnamesfirst]{natbib}
\usepackage{fullpage}
\usepackage{multirow, multicol,tasks,booktabs}
\usepackage{tikz}
\usetikzlibrary{arrows.meta,calc,positioning,shapes.geometric}
\usepackage{hyperref}
\usepackage[shortlabels]{enumitem}
\usepackage{todonotes}
\newtheorem{thm}{Theorem}[section]
\newtheorem{asm}{Assumption}

\newtheorem{prop}{Proposition}[section]
\newtheorem*{prop*}{Proposition}

\theoremstyle{definition}

\newtheorem{example}{Example}

\newtheorem*{not*}{Notations}
\usepackage{verbatim}
\usepackage{amsthm}

\allowdisplaybreaks

\title{Bounds for Standard Errors in Combined Data\thanks{We thank Isaiah Andrews, Pol Antrás, Michal Kolesár, Mikkel Plagborg-Møller, Elie Tamer, and participants of the 2023 New York Camp Econometrics, the 2025 Canadian Econometric Study Group, and the Princeton Econometrics Workshop for their valuable comments. This work is based on merging ``Bounds for Standard Errors from Interdependent Data'' by Jooyoung Cha and Yuya Sasaki with ``Top to Bottom: Best-case Standard Errors for Calibrated Model Parameters'' by Nelson Matthew P. Tan. All remaining errors are our own. Yuya Sasaki gratefully acknowledges generous funding from Brian and Charlotte Grove.}}
\author{
Jooyoung Cha\thanks{Jooyoung Cha: \texttt{jcha23@nd.edu}.
Department of Economics,
University of Notre Dame,
3060 Jenkins Nanovic Halls, 
Notre Dame, IN 46556} 
\\{\small Notre Dame}
\and 
Yuya Sasaki\thanks{Yuya Sasaki: \texttt{yuya.sasaki@vanderbilt.edu}.
Department of Economics,
Vanderbilt University,
VU Station B \#351819,
2301 Vanderbilt Place,
Nashville, TN 37235-1819}
\\{\small Vanderbilt}
\and
Nelson Matthew P. Tan\thanks{Nelson Matthew P. Tan: \texttt{matti.tan@princeton.edu}. 
Department of Economics, 
Princeton University, 
A85 Julis Romo Rabinowitz Building, 
Princeton, NJ 08544}
\\{\small Princeton}
}

\begin{document}
\maketitle
\begin{abstract}\setlength{\baselineskip}{7mm}
We propose methods for constructing lower bounds on the standard errors of parameters estimated from moment conditions obtained across different samples. 
Sharp explicit bounds are derived by exploiting geometric inequalities when no information about correlations across samples is available. 
Furthermore, we develop computationally tractable sharp bounds for more general settings with no or partial correlation information, which can be obtained by solving a simple semidefinite program. 
Finally, we illustrate the practical usefulness of our method through three empirical cases: two macroeconomics examples involving menu cost and Heterogeneous Agent New-Keynesian models; and a two sample instrumental variable microeconomic study. 
\end{abstract}

\newpage
\section{Introduction}\label{sec:intro}

Researchers frequently combine empirical moments from multiple, and potentially interdependent, data sources. 
Examples include integrating survey data with administrative records, combining cross-sectional data with time-series data, and using time-series observations recorded at different frequencies. 
In applications such as meta-analysis and calibration, the focus is often placed more on parameter estimation than on statistical inference. 
This emphasis may reflect the fact that conventional inference methods are typically infeasible, because the covariances of empirical moments across samples are often unknown, rendering standard errors for the parameter of interest impossible to compute.

When a covariance is not directly observable, it can still be bounded using marginal variances. 
Recent work by \citet{mikkel} derives a sharp upper bound on the covariance from marginal variances via the Cauchy--Schwarz inequality, and \citet{vohra2025inference} extends \citet{mikkel} along several dimensions. 
We complement this literature by characterizing the full range of standard errors that are feasible when only marginal variances are known. In addition to the sharp worst-case standard error, we derive the sharp best-case standard error and show that it is governed by a geometric cancellation condition. The resulting interval describes exactly how much uncertainty about the standard error is induced by the absence of covariance information.

The diagonal-only case provides the least-informative benchmark. The upper bound corresponds to maximal alignment of the sampling errors across empirical moments. The lower bound corresponds to maximal cancellation. This lower bound is often zero, but this should not be interpreted as a defect of the bound. Rather, it reveals when the unknown covariance structure can, in principle, offset the marginal sampling uncertainty in the combined estimator. Our characterization makes this possibility precise. With two moments, exact cancellation requires a knife-edge balance between the two components. With three or more moments, the condition has a geometric interpretation: the marginal contributions must be capable of closing into a polygon. Thus, the sharp lower bound describes the covariance configurations under which the combined estimator can have a very small standard error despite nontrivial uncertainty in its individual components.

The diagonal-only bounds are useful not because researchers should necessarily base inference on either extreme, but because they describe the full range of standard errors consistent with the available information. Figure~\ref{fig:correlation-decision-tree} outlines how applied researchers can use the bounds we derive and the existing knowledge of the data as a diagnostic for adjusting their empirical design.  If both the lower and upper bounds are large, the data are uninformative for the parameter of interest regardless of the unknown covariance structure. If the bounds are far apart, the missing covariance information is consequential: different admissible covariance structures lead to substantively different inferential conclusions. In that case, the bounds may help decide whether it is worth collecting, estimating, or justifying additional information about the covariance structure.

In many combined-data applications, researchers may know more than the marginal variances. The design of the data combination may imply that some components are estimated from independent samples, that two samples partially overlap, that dependence is confined within blocks, or that some pairwise correlations have known signs or plausible magnitude restrictions. We show how such information sharpens the standard error bounds. Information that rules out strong negative effective dependence raises the lower bound by limiting cancellation. Information that rules out strong positive effective dependence lowers the upper bound by limiting alignment. Exact independence, block independence, bounded sample overlap, and sign restrictions therefore have transparent effects on the range of feasible standard errors.

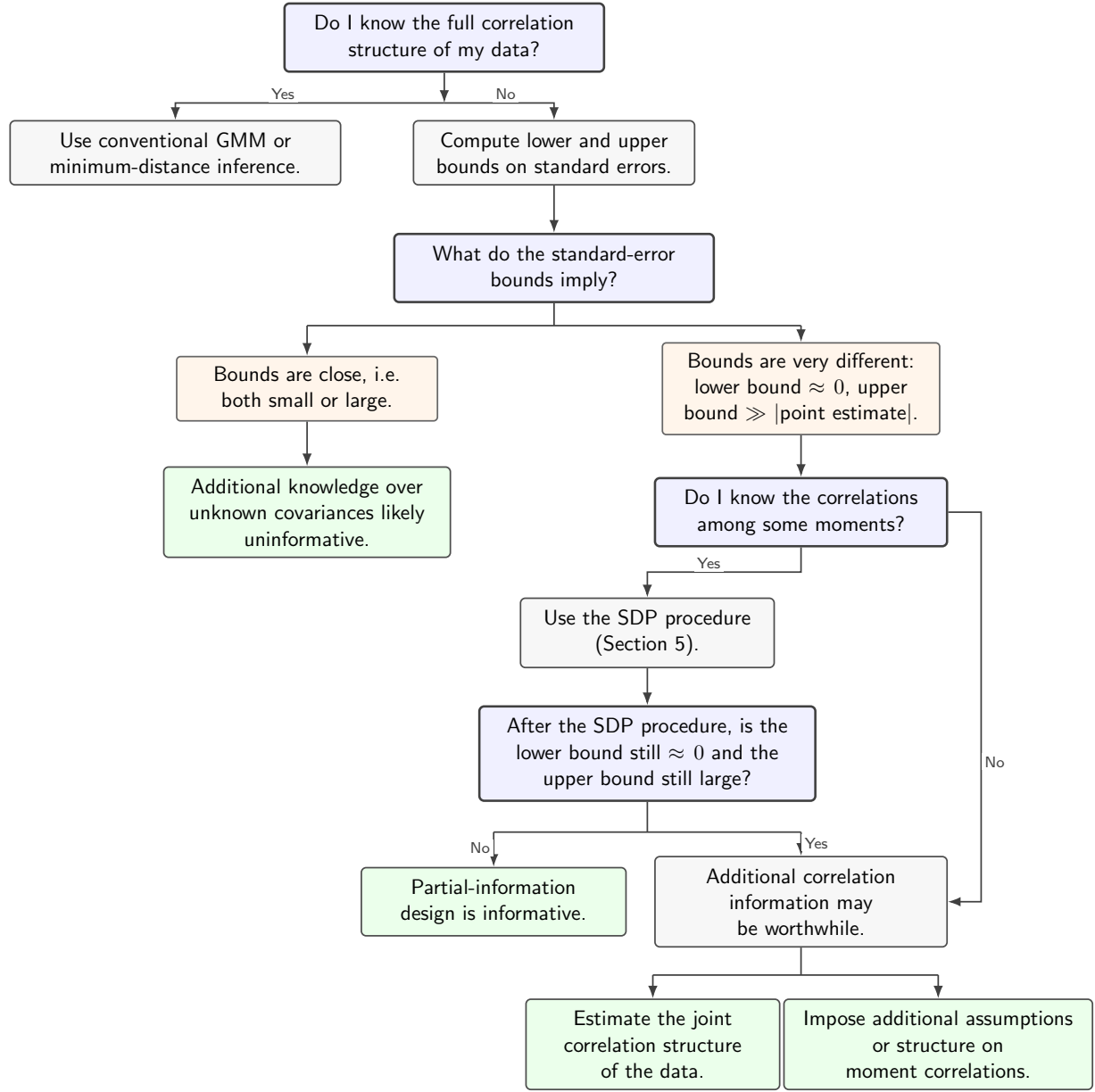
\begin{figure}[tbp]
  \centering
  \resizebox{0.98\textwidth}{!}{
%
\begin{tikzpicture}[
  font=\linespread{1}\sffamily\small,
  >=Latex,
  node distance=9mm and 10mm,
  decision/.style={
    rounded corners=2.5pt,
    draw=black!75,
    very thick,
    fill=blue!6,
    align=center,
    inner xsep=6pt,
    inner ysep=6pt,
    text width=5.4cm
  },
  proc/.style={
    rounded corners=2.5pt,
    draw=black!70,
    thick,
    fill=black!3,
    align=center,
    inner xsep=6pt,
    inner ysep=6pt,
    text width=4.7cm
  },
  case/.style={
    rounded corners=2pt,
    draw=black!60,
    thick,
    fill=orange!8,
    align=center,
    inner xsep=5pt,
    inner ysep=5pt,
    text width=4.0cm
  },
  outcome/.style={
    rounded corners=2.5pt,
    draw=black!70,
    thick,
    fill=green!8,
    align=center,
    inner xsep=6pt,
    inner ysep=6pt,
    text width=4.4cm
  },
  action/.style={
    rounded corners=2.5pt,
    draw=black!70,
    thick,
    fill=green!8,
    align=center,
    inner xsep=6pt,
    inner ysep=6pt,
    text width=4.4cm
  },
  edge/.style={->, thick, black!75},
  branch/.style={thick, black!75},
  lab/.style={font=\linespread{1}\sffamily\scriptsize, inner sep=1.5pt, fill=white}
]
 
\node[decision] (know) at (0,0) {Do I know the full correlation\\structure of my data?};
 
\node[proc, text width=5.6cm] (gmm) at (-4.9,-2.1) {Use conventional GMM or\\minimum-distance inference.};
 
\node[proc] (bounds) at (2.0,-2.1) {Compute lower and upper\\bounds on standard errors.};
\node[decision] (compare) at (2.0,-4.2) {What do the standard-error\\bounds imply?};
 
\node[case, text width=4.4cm] (closecase) at (-2.5,-6.4) {Bounds are close, i.e.\\both small or large.};
\node[outcome, text width=4.8cm] (closeoutcome) at (-2.5,-8.65) {Additional knowledge over\\unknown covariances likely\\uninformative.};
 
\node[case, text width=4.7cm] (differentcase) at (6.5,-6.4) {Bounds are very different:\\lower bound $\approx 0$, upper\\bound $\gg |\text{point estimate}|$.};
\node[decision, text width=4.9cm] (corr) at (6.5,-8.65) {Do I know the correlations\\among some moments?};
 
\node[proc, text width=4.2cm] (sdp) at (3.7,-10.85) {Use the SDP procedure\\(Section~\ref*{sec:SDP}).};
\node[decision, text width=5.7cm] (after) at (3.7,-13.05) {After the SDP procedure, is the\\lower bound still $\approx 0$ and the\\upper bound still large?};
 
\node[outcome, text width=4.4cm] (partial) at (0.9,-15.75) {Partial-information\\design is informative.};
\node[proc, text width=4.9cm] (worth) at (6.5,-15.75) {Additional correlation\\information may be worthwhile.};
 
\node[action, text width=4.2cm] (estimate) at (3.8,-18.35) {Estimate the joint\\correlation structure\\of the data.};
\node[action, text width=5.2cm] (assume) at (9.0,-18.35) {Impose additional assumptions\\or structure on\\moment correlations.};
 
\coordinate (topSplit) at ($(know.south)+(0,-0.55)$);
\draw[edge] (know.south) -- (topSplit);
\draw[edge] (topSplit) -| node[lab,pos=.30,above] {Yes} (gmm.north);
\draw[edge] (topSplit) -| node[lab,pos=.28,above] {No} (bounds.north);
 
\draw[edge] (bounds.south) -- (compare.north);
 
\coordinate (caseStem) at (2.0,-5.25);
\coordinate (caseLeft) at (-2.5,-5.25);
\coordinate (caseRight) at (6.5,-5.25);
\draw[branch] (compare.south) -- (caseStem);
\draw[branch] (caseLeft) -- (caseRight);
\draw[edge] (caseLeft) -- (closecase.north);
\draw[edge] (caseRight) -- (differentcase.north);
 
\draw[edge] (closecase.south) -- (closeoutcome.north);
\draw[edge] (differentcase.south) -- (corr.north);
 
\draw[edge] (corr.south) -- ++(0,-0.45) -| node[lab,pos=.30,above] {Yes} (sdp.north);
\draw[edge] (corr.east) -- ++(0.6,0) |- node[lab,pos=.32,right] {No} (worth.east);
 
\draw[edge] (sdp.south) -- (after.north);
 
\coordinate (afterHub) at ($(after.south)+(0,-0.55)$);
\coordinate (afterLeft) at (partial.north |- afterHub);
\coordinate (afterRight) at (worth.north |- afterHub);
\draw[branch] (after.south) -- (afterHub);
\draw[branch] (afterLeft) -- (afterRight);
\draw[edge] (afterLeft) -- node[lab,pos=.45,left] {No} (partial.north);
\draw[edge] (afterRight) -- node[lab,pos=.45,right] {Yes} (worth.north);
 
\coordinate (worthHub) at ($(worth.south)+(0,-0.50)$);
\coordinate (estimateHub) at (estimate.north |- worthHub);
\coordinate (assumeHub) at (assume.north |- worthHub);
\draw[branch] (worth.south) -- (worthHub);
\draw[branch] (estimateHub) -- (assumeHub);
\draw[edge] (estimateHub) -- (estimate.north);
\draw[edge] (assumeHub) -- (assume.north);
 
\end{tikzpicture}}
  \caption{Decision tree for choosing an empirical strategy under partial information about the correlation structure.}
  \label{fig:correlation-decision-tree}
\end{figure}

We develop two complementary approaches. First, we provide analytic characterizations for several empirically relevant forms of covariance information. For example, when moments can be partitioned into independent blocks, the bounds decompose block by block. When two components are estimated from partially overlapping samples, the overlap structure can place direct limits on the magnitude of their correlation. These cases show how features of the data-combination design translate into sharper inference. Second, for more general forms of partial covariance information, we formulate the problem as a semidefinite program (SDP). This computational approach allows researchers to incorporate arbitrary linear or convex restrictions on the unknown covariance matrix and obtain the corresponding standard-error bounds using standard numerical solvers.

Our framework applies to a wide range of empirical settings. In two-sample instrumental variables, conventional inference often relies on an independence assumption between the sample used to estimate the first stage and the sample used to estimate the reduced form. Our approach nests that assumption but also allows researchers to examine weaker alternatives, such as bounded cross-sample dependence. Similar issues arise in empirical macroeconomics, where constructed or generated shock series are often estimated in one sample and then used in downstream analyses. Classical generated-regressor arguments, such as those associated with \cite{pagan1984econometric}, justify standard inference when the relevant covariance structure is negligible or consistently estimable. But when a constructed shock is made externally and then used in a different estimating sample, the relevant joint covariance may not be directly recoverable. Moreover, calibration, the most common setting in macroeconomics applicable to our framework \citep{kaplan&violante:2014, mckayetal:2016, nakamura&steinsson:2018, kydland&prescott:1982, gourinchas&parker:2003, arellano:2008, eatonetal:2011}, faces an analogous problem under an unknown covariance structure among moments used for matching. Our bounds then quantify the range of standard errors consistent with the available marginal information.

The paper contributes to the broader literature on data combination, where much of the existing work focuses on identification and estimation rather than inference; see \cite{handbook}. It also contributes to recent work on inference when covariance information is incomplete, especially \cite{mikkel} and \cite{vohra2025inference}. Relative to this literature, our contribution is threefold. We derive sharp lower as well as upper bounds on standard errors under diagonal-only information. We characterize the geometry of the lower bound and the conditions under which exact cancellation is possible. Finally, we show how partial information about the covariance structure moves the identified interval from the diagonal-only benchmark toward the usual full-information standard error.

The remainder of the paper is organized as follows. 
Section~\ref{sec:framework} introduces the framework. 
Section~\ref{sec:analytic} derives the explicit sharp bounds. Section~\ref{sec:cov_info} characterizes the set of achievable standard errors and how additional information on the covariances of the moments can tighten the bounds on this set. Section~\ref{sec:SDP} presents the SDP approach. Section~\ref{sec:empirics} illustrates empirical applications. 
Section~\ref{sec:summary} concludes. 
All mathematical proofs are collected in the appendix.

\section{The Framework}\label{sec:framework}
Suppose that we are interested in a scalar parameter $\varphi(\theta)$, where $\varphi : \mathbb{R}^k \to \mathbb{R}$ is a known function of a $k$-dimensional vector of structural parameters $\theta$. 
The structural parameters, in turn, determine a $p$-dimensional vector of reduced-form parameters $\mu = h(\theta)$ through a known injection $h : \mathbb{R}^k \to \mathbb{R}^p$.

If estimates $\hat{\mu} = (\hat{\mu}_1, \cdots, \hat{\mu}_p)'$ of the reduced-form parameters $\mu$ are available, then standard inference about $\varphi(\theta)$ is based on the linear approximation
\begin{align}\label{eq:phi_linear}
    \varphi(\hat{\theta}) - \varphi(\theta) = \ell' (\hat{\mu} - \mu) + o_p(n^{-1/2}).
\end{align}
The results below rely only on the linear representation in \eqref{eq:phi_linear}. The structural mapping $h$ provides one leading way to obtain the loading vector $\ell$, but the same framework applies to any estimator whose scalar target admits such a first-order expansion, including moment-matching and two-step estimators discussed below.

In the exactly identified case in which $\hat\theta=h^{-1}(\hat\mu)$, the loading vector is
\[
\ell =
\frac{\partial \varphi}{\partial \theta'}(\theta)
\cdot
\frac{\partial h^{-1}}{\partial \mu}(\mu).
\]

In particular, if we have the standard convergence in distribution
\begin{align}\label{eq:mu_asym}
    \sqrt{n} (\hat{\mu} - \mu) \overset{d}{\to} N(0, \Sigma),
\end{align}
then the delta method yields
\begin{align}\label{eq:phi_asym}
    \sqrt{n}(\varphi(\hat{\theta}) - \varphi(\theta)) \overset{d}{\to} N(0, \sigma^2), \\
    \text{where} \quad \sigma^2 = \ell' \Sigma \ell. \notag
\end{align}
This asymptotic distribution facilitates inference about $\varphi(\theta)$, provided that the asymptotic variance $\sigma^2$ is estimable.

However, if the reduced-form estimates $\hat{\mu}_1, \cdots, \hat{\mu}_p$ are obtained from different samples, then not all elements of the asymptotic covariance matrix $\Sigma$ are estimable. 
To fix ideas, suppose that $\hat{\mu}_j$ is computed from a sample $X_j = (X_{j,1}, \cdots, X_{j,n_j})$, where each $X_j$ is observed separately for $j = 1, \cdots, p$. 
In this setting, only the diagonal elements of $\Sigma$ are estimable:
\[
s_j^2 := \Sigma_{jj}, \quad \text{for } j = 1, \cdots, p.
\]
With the off-diagonal elements of $\Sigma$ unknown, however, it is not possible to estimate the asymptotic variance $\sigma^2 = \ell' \Sigma \ell$ of the parameter $\varphi(\theta)$.

\bigskip
\noindent
\begin{example}
    Many economic questions can be formulated within this framework of data combination. 
    For example, consider the moment-matching estimator
    \begin{align*}
        \hat{\theta} = \underset{\theta \in \Theta}{\arg\min} \  (\hat{\mu} - \varphi(\theta))' \hat{W} (\hat{\mu} - \varphi(\theta)),
    \end{align*}
    where $\Theta$ is a compact parameter space. 
    The loading vector $\ell$ can be estimated by 
    \begin{align}
    \label{eq:min_distance_loadings}
    &\hat{\ell} = \hat{W} \hat{G} (\hat{G}' \hat{W} \hat{G})^{-1}\hat{\lambda}
    \\
    &\text{with} \quad \hat{G} \equiv \frac{\partial h(\hat{\theta})}{\partial \theta'} \text{ and } \hat{\lambda} \equiv \frac{\partial \varphi(\hat{\theta})}{\partial \theta}.
    \notag
    \end{align}
    However, if the moments $\hat{\mu}_1, \cdots, \hat{\mu}_p$ are obtained from separate samples $X_1, \cdots, X_p$, the covariance matrix $\Sigma$ is not estimable. 
    In particular, while $\hat{\ell}$ is estimable, the asymptotic variance $\sigma^2 = \ell' \Sigma \ell$ of $\varphi(\theta)$ remains unidentified due to the unobservability of the off-diagonal elements of $\Sigma$.
    $\blacktriangle$
\end{example}

\bigskip
\noindent
\begin{example}
    Our framework encompasses two-step estimation, where an ancillary parameter $\gamma$ is first estimated from a sample $\{Z_{j}\}_{j=1}^{n_{1}}$, and the resulting estimate $\hat\gamma$ is then used to estimate the parameter of interest $\beta$ using another sample $\{X_{i}\}_{i=1}^{n_{2}}$. Our framework becomes essential when the score for $\beta$ based on $\{X_{i}\}_{i=1}^{n_{2}}$ has unknown correlations with the first-stage estimate $\hat\gamma$ which is constructed from the data $\{Z_{j}\}_{j=1}^{n_{1}}$.

    For instance, consider a two-step estimator of the form 
    \begin{equation*}
        \hat{\beta} = \underset{\beta \in B}{\arg\max}\frac{1}{n_{2}}\sum_{i=1}^{n_{2}}f(X_{i},\beta,\hat{\gamma}),
    \end{equation*}
    where $\hat{\gamma}$ is obtained via some first stage estimator on $\{Z_{j}\}_{j=1}^{n_{1}}$. Under regularity conditions, one can rewrite $\hat{\beta}$ as the solution to the moment condition 
    \begin{equation*}
        \frac{1}{n_{2}}\sum_{i=1}^{n_{2}}\frac{\partial}{\partial \beta}f(X_{i},\beta,\hat{\gamma}) = 0
    \end{equation*}
    with a similar reformulation for $\hat{\gamma}$. Collecting these orthogonality-type conditions and parameters into vectors $\hat{\mu}$ and $\theta = (\gamma, \beta)'$, respectively, we obtain exactly the same setting as in the previous example.

    Practitioners using two-step estimation with separate data often compute standard errors via the bootstrap. However, bootstrapping with independent resampling across the two samples implicitly assumes that the two datasets are independent. If this assumption is not warranted, the resulting bootstrap standard errors may be misleading.
    $\blacktriangle$
\end{example}
\bigskip

In the remainder of the paper, we propose methods to obtain sharp lower bounds for the asymptotic variance $\sigma^{2}$ of interest, using the marginal variances $s_{j}^{2}$ for $j = 1, \dots, p$, together with no or only partial information about the covariances.

\section{Explicit Sharp Bounds}\label{sec:analytic}

If the researcher has information only about the diagonal elements of $\Sigma$, then explicit sharp bounds for $\sigma$ can be derived. 
The following theorem formally states this result.

\begin{thm}\label{thm:explicit}
Under the framework \eqref{eq:phi_linear}--\eqref{eq:phi_asym}, the sharp lower and upper bounds for $\sigma$, when the researcher knows only the diagonal elements of $\Sigma$, are explicitly given by
\begin{align*}
    \max\left\{\max_{m}\left\{|\ell_{m}|s_{m} - \sum_{j \neq m}|\ell_{j}|s_{j}\right\},0\right\} \qquad\text{and}\qquad \sum_{j=1}^p |\ell_j|s_j,
\end{align*}
respectively.
\end{thm}

\noindent
Proof of this statement is found in Appendix \ref{sec:thm:explicit}.

The sharp upper bound coincides with the result of \citet{mikkel}; this component of Theorem \ref{thm:explicit} is therefore not novel, and we include it solely for completeness and the convenience of the reader. In contrast, the sharp lower bound established in Theorem \ref{thm:explicit} is new to the literature. 

By the continuous mapping theorem, these sharp bounds can be consistently estimated under the following assumption.

\begin{asm}\label{asm:sufficient}
$
\hat{\ell}_j \overset{p}{\to} \ell_j
$    
and
$
\hat{s}_j \overset{p}{\to} s_j
$    
for all $j \in \{1,\cdots,p\}$. 
\end{asm}

\begin{prop}\label{prop:explicit_consistency}
Under Assumption \ref{asm:sufficient}, the sharp lower and upper bounds of $\sigma$, when the researcher knows only the diagonal elements of $\Sigma$, are consistently estimated by
\begin{align*}
    \max \left\{\max_{m}\left\{|\hat \ell_m|\hat s_m - \sum_{j\neq m}|\hat \ell_j|\hat s_j\right\}, 0 \right\} \qquad\text{and}\qquad \sum_{j=1}^p |\hat \ell_j|\hat s_j, \qquad\text{respectively.}
\end{align*}
\end{prop}
The analytical bounds can be interpreted as the following.
The upper bound is attained when the estimation errors associated with all moments are perfectly aligned in the direction that increases the variance of the scalar target. The lower bound is attained when these components cancel as much as possible. With two moments, exact cancellation is possible if and only if $|\ell_{1}|s_{1} = |\ell_{2}|s_{2}$. More generally, the lower bound is zero if and only if the largest $|\ell_j| s_j$ is no larger than the sum of the remaining $|\ell_k|s_k$'s. This is the polygon condition: the lengths $|\ell_1|s_1,\dots,|\ell_p|s_p$ can be arranged as vectors whose sum is zero. Thus, a zero lower bound reflects feasible covariance cancellation under diagonal-only information, not a failure of the bound.

The analytical lower bound provided by Theorem \ref{thm:explicit} and Proposition \ref{prop:explicit_consistency} allows all correlation structures consistent with positive semidefiniteness. In many applications, however, the researcher may have some information about the signs, magnitudes, or block structure of the unknown correlations. Such information restricts the covariance structures that can generate maximal cancellation or maximal alignment. The subsequent sections formalize this idea.

\section{Covariance Information and Bound Tightening}\label{sec:cov_info}
This section introduces empirically plausible covariance information and how it may tighten the bounds in Section \ref{sec:analytic}.
We show that the lower bound is driven by the extent to which the unknown covariance structure can generate cancellation across empirical moments, whereas the upper bound is driven by the extent to which it can generate alignment. This representation also clarifies how additional information about the covariance matrix tightens the bounds.

\subsection{The Range of Achievable Standard Errors}\label{sec:range}

Before investigating the role of covariance information, we first characterize the set of achievable standard errors without information restriction. Furthermore, we study where the best-case and worst-case correlation structures are located in the set of feasible correlation matrices.

Proposition~\ref{prop_se_range} establishes that the set of achievable standard errors takes the form of a simple interval, with no gaps between its endpoints. Consequently, the full-information standard errors, those 
computed with knowledge of the off-diagonal elements of $\hat{\Sigma}$, can be expressed as a weighted average of the best-case and worst-case standard errors.

\begin{prop}
    \label{prop_se_range}
The set of achievable standard errors forms a closed interval with two endpoints given by the best-case and worst-case standard errors. In particular, the full-information standard errors satisfy
\begin{align}
    \mathrm{SE}_{\mathrm{F}}(\theta) 
    &= \alpha \, \mathrm{SE}_{\mathrm{BCS}}(\theta) 
       + (1 - \alpha)\, \mathrm{SE}_{\mathrm{WCS}}(\theta),
    \quad \alpha \in [0,1],
\end{align}
where $\mathrm{SE}_{\mathrm{F}}(\theta)$ denotes the full-information standard errors, $\mathrm{SE}_{\mathrm{BCS}}(\theta)$ the best-case standard errors, and $\mathrm{SE}_{\mathrm{WCS}}(\theta)$ the worst-case standard errors.

\end{prop}

\noindent
Proof of this statement is found in Appendix \ref{proof_prop_range}. 

We can also generally characterize where the minimal and maximal correlation structures lie within the set of feasible correlation matrices. In the special case of two moments, the best-case correlation occurs when the two moments are either perfectly positively or perfectly negatively correlated.

\begin{prop}
    \label{prop_structure_location}
The best-case correlation matrix, $\Omega_{\mathrm{BCS}}$, and the worst-case correlation matrix, $\Omega_{\mathrm{WCS}}$, lie on the boundary of the set of valid correlation matrices
\[
    C = \bigl\{ \Omega : \Omega \succeq 0,\ \Omega_{i,i} = 1 \ \ \forall i = 1,\cdots,p \bigr\}.
\]
When $p = 2$, the best-case correlation $\rho_{\mathrm{BCS}}$ between the two empirical moments depends on the sign of the product of the entries of $\ell$, so that 
\begin{align}
   \rho_{\mathrm{BCS}} = 
    \begin{cases}
        -1, & \text{if } \mathrm{sgn}(\ell_{1}\ell_{2}) > 0, \\[6pt]
         1, & \text{if } \mathrm{sgn}(\ell_{1}\ell_{2}) < 0,
    \end{cases}
\end{align}
and the opposite holds for the worst-case correlation, $\rho_{\mathrm{WCS}}$.
\end{prop}

\noindent
Proof of this statement is found in Appendix \ref{proof_prop_location}.

In the $p = 2$ case of Proposition~\ref{prop_structure_location}, the dependence on the sign of 
the product of the entries of $\ell$ highlights the mechanism that drives the best-case 
correlation away from the worst-case correlation.

\subsection{Alignment and cancellation}
Let $D = \operatorname{diag}(s_1,\ldots,s_p)$ where $s_j^2=\Sigma_{jj}$ for $j=1,\cdots,p$.
Write $\Sigma = D R D$,
where $R$ is a correlation matrix. 
Define $z_j = |\ell_j|s_j$ for each $j = 1,\cdots,p$, and collect them as the vector $z=(z_1,\ldots,z_p)'$.
Let $S$ be the diagonal matrix with $j$-th diagonal entry equal to
$\operatorname{sign}(\ell_j)$, with an arbitrary sign convention when $\ell_j=0$, and define
the sign-adjusted correlation matrix
$T = SRS$.
Then, $T$ is also a correlation matrix, and the asymptotic variance can be written as
\begin{align}\label{eq:sigma_with_T}
    \sigma^2 = \ell' \Sigma \ell = z'T z = \sum_{j=1}^p z_j^2 + 2\sum_{i<j} T_{ij} z_i z_j.
\end{align}

This representation is useful because the sign of $T_{ij}$ has a direct interpretation. 
When $T_{ij}>0$, the estimation errors associated with moments $i$ and $j$ are effectively aligned in the direction relevant for $\varphi(\theta)$, increasing the variance. When $T_{ij}<0$, they are effectively offsetting each other, decreasing the variance. Thus, restrictions that prevent $T_{ij}$ from being very negative limit cancellation and tend to raise the lower bound, while restrictions that prevent $T_{ij}$ from being very positive limit alignment and tend to lower the upper bound.

This distinction is important because the raw correlation $R_{ij}$ need not have the same interpretation as the effective correlation $T_{ij}$. If $\ell_i$ and $\ell_j$ have the same sign, then $T_{ij}=R_{ij}$. If they have opposite signs, then $T_{ij}=-R_{ij}$. Thus, a positive raw correlation may induce cancellation when the two moments enter the scalar target with opposite signs.

Let $\mathcal T$ denote a nonempty set of admissible sign-adjusted correlation matrices. For example, $\mathcal T$ may encode diagonal-only information, known zero correlations, sign restrictions, magnitude restrictions, block restrictions, or any combination of these. Define the lower and upper feasible standard errors associated with $\mathcal T$ by
\begin{align*}
    \underline{\sigma}(\mathcal T)
        =
        \left(
        \inf_{T\in\mathcal T} z'Tz
        \right)^{1/2} 
        \qquad\text{and}\qquad
        \overline{\sigma}(\mathcal T)
        =
        \left(
        \sup_{T\in\mathcal T} z'Tz
        \right)^{1/2},
\end{align*}
respectively.
When $\mathcal T$ is the full set of correlation matrices, these endpoints coincide with the diagonal-only bounds in Theorem \ref{thm:explicit}.
This notation makes the role of covariance information explicit. If $\mathcal T_2\subseteq\mathcal T_1$, then
\begin{align*}
\underline{\sigma}(\mathcal T_1) \leq
\underline{\sigma}(\mathcal T_2)
\qquad\text{and}\qquad
\overline{\sigma}(\mathcal T_2)\leq \overline{\sigma}(\mathcal T_1).
\end{align*}
That is, additional covariance information weakly contracts the feasible range of standard errors.
The rest of this section studies which empirical forms of covariance information shrink $\mathcal T$ in ways that are informative for the lower bound, the upper bound, or both.

\subsection{How covariance information tightens the bounds}

The representation in \eqref{eq:sigma_with_T} makes clear which types of covariance information affect which
endpoint of the range. Negative values of $T_{ij}$ generate cancellation, while positive values generate alignment. Consequently, restrictions that bound $T_{ij}$ away from $-1$ limit cancellation and can
raise the lower bound. Restrictions that bound $T_{ij}$ away from $1$ limit alignment and
can lower the upper bound.
A useful general form of pairwise covariance information is 
\begin{align*}
    \underline{\rho}_{ij}\leq T_{ij}\leq \overline{\rho}_{ij},
\end{align*}
where the lower endpoint $\underline{\rho}_{ij}$ restricts the amount of cancellation available between moments $i$ and $j$, while the upper endpoint $\overline{\rho}_{ij}$ restricts the amount of alignment available between them. Thus, a one-sided lower restriction is primarily informative for the best-case standard error, a one-sided upper restriction is primarily informative for the worst-case standard error, and a two-sided restriction may tighten both.

\subsection{Empirical sources of covariance information}
We discuss several empirically motivated restrictions.

Sign restrictions are an important special case. If $T_{ij}\ge0$, then the pair $(i,j)$ cannot contribute to cancellation, although it may still contribute to alignment. Such information is therefore informative for the lower bound. If $T_{ij}\le0$, then the pair cannot contribute to alignment, although it may still contribute to cancellation; this is informative for the upper bound. If $T_{ij}=0$, then the pair contributes neither to cancellation nor to alignment, and both endpoints may tighten. These restrictions are stated in terms of the sign-adjusted correlation $T$, not the raw correlation $R$; the two coincide only when $\ell_i$ and $\ell_j$ have the same sign.

These restrictions often arise from the design of a combined-data problem. If two moments are estimated from independent samples, the corresponding covariance is zero and hence $T_{ij}=0$. If moments can be partitioned into independent blocks, then all cross-block entries of $T$ are zero, while dependence may remain unrestricted within each block. If moments are estimated from partially overlapping samples, then the overlap structure may imply a magnitude bound of the form
\begin{align*}
    |T_{ij}|\leq q_{ij},
\end{align*}
where $q_{ij}<1$ summarizes the maximal cross-source dependence allowed by the sampling
design. More generally, weak dependence across markets, cohorts, time periods, or data
sources can be represented by similar bounds on cross-block entries of $T$.

In a few special cases, these restrictions yield closed-form bounds. For example, with two moments and $T_{12}\in[\underline{\rho},\overline{\rho}]$, we have
\begin{align*}
    \underline{\sigma}
        =
        \left(z_1^2+z_2^2+2\underline{\rho}z_1z_2\right)^{1/2}
        \qquad\text{and}\qquad
        \overline{\sigma}
        =
        \left(z_1^2+z_2^2+2\overline{\rho}z_1z_2\right)^{1/2}.
\end{align*}
This formula nests independence, sign restrictions, and bounded-dependence restrictions. For instance, independence corresponds to $\underline{\rho}=\overline{\rho}=0$, while $|T_{12}|\le q$ corresponds to
$[\underline{\rho},\overline{\rho}]=[-q,q]$.

Another tractable case is block independence. Suppose that the moments are partitioned into blocks, cross-block correlations are zero, and within each block no off-diagonal information is available. Then, the diagonal-only formulas of Theorem~\ref{thm:explicit} apply within each block, and the resulting block-level variances add across blocks. Thus, block information rules out cross-block cancellation and alignment while preserving the diagonal-only benchmark within blocks.

\subsection{Illustration: zero cancellation with three moments}
The preceding discussion concerns the full range of feasible standard errors. It is also useful to visualize when additional covariance information rules out exact cancellation, that is, when it rules out a zero lower bound. Consider the case of three moments and normalize the marginal contributions by  $q_j=z_j/(z_1+z_2+z_3)$. Each point in the simplex represents the relative importance of the three marginal
contributions. Under diagonal-only information, exact cancellation is feasible if and only if $\max_j q_j\le1/2$. 
If exact cancellation occurs, then the sign-adjusted correlations must satisfy $Tq=0$, where $$T = \left( \begin{array}{ccc}
     1& t_{12}&t_{13} \\
     t_{12}&1&t_{23}\\
     t_{13}&t_{23}&1
\end{array}\right)
\qquad\text{and}\qquad
q=\left(\begin{array}{c}q_1 \\ q_2 \\ q_3\end{array}\right).$$
Solving this linear system gives the following pairwise correlations:
\begin{align*}
    T_{ij}^{0}(q)
        =\frac{q_k^2-q_i^2-q_j^2}{2q_iq_j} \qquad \text{for all } \{i,j,k\}=\{1,2,3\}.
\end{align*}
If the researcher imposes $T_{ij}\ge-\gamma$ for all $i,j$, then exact cancellation requires
\begin{align*}
    \max_j q_j\leq \frac{1}{2} \quad\text{and}\quad T_{ij}^{0}(q)\geq-\gamma \quad\text{for all } i\neq j.
\end{align*}
As $\gamma$ decreases, this zero-cancellation region shrinks. Figure \ref{fig:zero_lb_regions} plots this
region for different values of $\gamma$. The simplex describes the relative sizes of the three marginal contributions. 
\begin{figure}[tb!]
    \begin{center}
    \includegraphics[width=\linewidth]{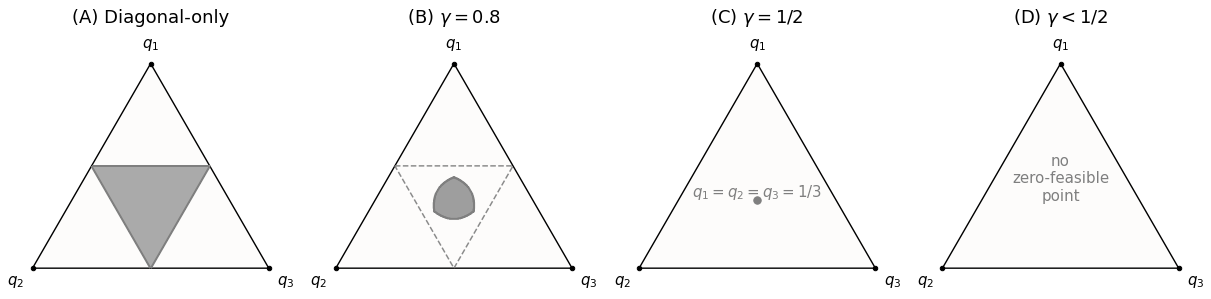}
    \caption{Zero-cancellation regions under bounded negative effective correlation}\label{fig:zero_lb_regions}      
    \end{center}
{\footnotesize {Note:} Zero-cancellation regions in the three-moment case. Each point in the simplex corresponds to normalized marginal contributions $q_j=z_j/\sum_k z_k$. The shaded region contains the values of $q$ for which exact cancellation is feasible, and hence the sharp lower bound on the standard error is zero. Outside the shaded region, exact cancellation is ruled out and the lower bound is strictly positive. The first panel corresponds to diagonal-only information. The remaining panels impose $T_{ij}\ge -\gamma$ for all pairs, which limits the amount of negative effective correlation. As $\gamma$ decreases, the zero-cancellation region shrinks. The figure displays zero feasibility, not the full range of feasible standard errors.}
\end{figure}
Shaded points are those for which the available covariance information still permits exact cancellation. Unshaded points are those for which the zero lower bound is no longer feasible.

These examples illustrate the general mechanism. Lower-bound informativeness comes from ruling out the negative effective correlations required for cancellation. Upper-bound informativeness comes from ruling out the positive effective correlations required for alignment. The next section describes how these restrictions can be imposed computationally when closed-form expressions are not available.

\section{Semidefinite Program}\label{sec:SDP}

The previous section shows that a researcher can obtain explicit sharp bounds when the values of only the diagonal elements of $\Sigma$ are known.
In practice, however, there are situations where the researcher also possesses partial information about the off-diagonal elements of $\Sigma$.
For example, block-diagonal elements of $\Sigma$ are known when multiple moment conditions are constructed from each data set.

We show that, in such general settings, the problem of obtaining bounds for the standard error can be reformulated as an equivalent semidefinite program (SDP), which can be readily implemented using existing numerical solvers.
Since \citet{mikkel} analyzes worst-case standard errors, we focus hereafter on the best-case standard errors.


\subsection{The Best-Case Standard Errors}\label{bcs:derivation}

Let $\mathcal{I} \subset \mathbb{N}^2$ denote the set of row--column index pairs of the variance estimate $\hat{\Sigma}$ for which the researcher possesses information. 
For example, if the researcher knows only the diagonal elements of $\hat{\Sigma}$, then 
$\mathcal{I} = \{(1,1), \cdots, (p,p)\}$. 
If, in addition, the researcher knows certain off-diagonal elements $\hat{\Sigma}_{jk}$ and $\hat{\Sigma}_{kj}$, then the corresponding index pairs $(j,k)$ and $(k,j)$ are also included in $\mathcal{I}$.
Define $S(\hat{\Sigma}, \mathcal{I})$ as the set of all symmetric, positive semidefinite matrices that share with $\hat{\Sigma}$ the entries corresponding to the index pairs in $\mathcal{I}$, that is,
\[
S(\hat{\Sigma}, \mathcal{I})
\equiv
\left\{
\Tilde{\Sigma} \succeq 0 \ :\ 
\Tilde{\Sigma} = \Tilde{\Sigma}', \ 
\Tilde{\Sigma}_{ij} = \hat{\Sigma}_{ij} \text{ for all } (i,j) \in \mathcal{I}
\right\},
\]
where $A \succeq 0$ indicates that the matrix $A$ is positive semidefinite.

With the above notation, our problem now concerns the computation of $n^{-1/2}$ times
\begin{equation}\label{eq:pre_sdp_opt}
\min_{\Tilde{\Sigma} \in S(\hat\Sigma,\mathcal{I})} \sqrt{\hat{\ell}'\Tilde{\Sigma}\hat{\ell}},
\end{equation}
To formally motivate this minimization problem, we state the following assumption.
\begin{asm}\label{asm:sufficient2}
(i) $(j,j) \in \mathcal{I}$ for all $j \in \{1,\cdots,p\}$.
(ii)
$
\hat{\ell}_j \to \ell_j
$    
for all $j \in \{1,\cdots,p\}$,
and
$
\hat{\Sigma}_{jk} \to \Sigma_{jk}
$    
for all $(j,k) \in \mathcal{I}$. 
\end{asm}

\noindent
Part~(i) requires that all diagonal elements of $\hat{\Sigma}$ are known, consistent with the setting in the previous section. 
Part~(ii) parallels Assumption~\ref{asm:sufficient} and is extended to the current framework. 

The following proposition formally justifies \eqref{eq:pre_sdp_opt} as the main object of interest.

\begin{prop}\label{prop:consistency}
Under Assumption \ref{asm:sufficient2}, Equation \eqref{eq:pre_sdp_opt} is consistent for $\min_{\Tilde\Sigma \in S(\Sigma,\mathcal{I})} \sqrt{\ell'\Tilde\Sigma\ell}$.
\end{prop}

\noindent
Proof of this statement is found in Appendix \ref{sec:prop:consistency}.

Furthermore, the following proposition establishes that the constrained optimization problem \eqref{eq:pre_sdp_opt} can be reformulated as a numerically feasible SDP.

\begin{prop}\label{prop_sdp_solution}
The solution to the optimization problem \eqref{eq:pre_sdp_opt} is also the solution to the SDP:
        \begin{align}
        \min_{\Omega} \ & \hat{\ell}' \hat D\Omega \hat D \hat{\ell} \label{sdp}\\
        \text{s.t. } &\Omega \succeq 0,\label{sdp_constraint}\\\
        &\Omega_{jk} = \hat{\rho}_{jk} \ \forall (j,k) \in \mathcal{I},\label{sdp_affine_corr}
        \end{align} 
where $\hat D \equiv \mathrm{Diag}(\hat{\Sigma})^{\frac{1}{2}}$ and $\hat\rho_{jk} \equiv \hat\Sigma_{jk} \left/ \sqrt{\hat\Sigma_{jj} \hat\Sigma_{kk}} \right.$ for all $(j,k) \in \mathcal{I}$.
\end{prop}

\noindent
Proof of this statement is found in Appendix \ref{sec:appendix1.1}.

\subsection{Implementation}\label{sec:implementation}

With Proposition \ref{prop_sdp_solution}, we can obtain the best-case standard error \eqref{eq:pre_sdp_opt} by solving the SDP \eqref{sdp}--\eqref{sdp_affine_corr} after computing $\hat\ell$, $\hat\Sigma$, and $\hat{\rho}_{ij}$ for every pair of moments $i$ and $j$ whose correlation we can estimate. 

Although reliable SDP solvers are available in virtually every computing environment,\footnote{For this paper, the preferred language was Python, and the SDP solver employed was the function provided in the package \texttt{CVXOPT}.} numerical optimization may occasionally yield undesired results. In particular, depending on the numerical properties of the problem or the available computational resources, applying an SDP solver to \eqref{sdp}--\eqref{sdp_affine_corr} may produce an $\Omega$ that is not exactly a correlation matrix. For instance, the diagonal elements of $\Omega$ may deviate from $1$, or the off-diagonal elements $\Omega_{j,k}$ may fall outside the interval $[-1,1]$.

Nonetheless, such violations occur only at a negligible scale,\footnote{For example, one of the diagonal entries may take the value $\Omega_{j,j} = 0.9999999$, as observed in the initial computations of the best-case standard errors for the parameter $n$ of the \citet{alvarez&lippi:2014} menu cost model.} and are more likely attributable to limitations of the software than to any substantive theoretical flaw. To address this issue, a researcher may round the entries of the resulting correlation matrix, which yields virtually identical standard errors for the estimated parameters. Rather than directly rounding the entries of the $\Omega$ produced by the SDP solver, the researcher can proceed as follows: first, run the SDP; then, if the solution does not correspond to a valid correlation matrix, apply the following alternative method to construct a correlation matrix that satisfies \eqref{sdp}--\eqref{sdp_affine_corr}:

\begin{enumerate}
    \item Compute the eigenvalue decomposition of the initial correlation matrix output $\Omega_{0}$. Letting $\Lambda_{0}$ be the diagonal matrix of eigenvalues and $E_{0}$ be the $p \times p$ matrix of corresponding eigenvectors, we have 
    $\Omega_{0} = E_{0}\Lambda_{0}E_{0}.$
    \item Round each eigenvalue up such that negative eigenvalues become $0$.\footnote{Eigenvalues such as $-0.000000398$ (taken from the correlation matrix referred to in the previous footnote), for instance, are rounded to $0$. From the applications done in this paper, the numerical optimization outputs correlation matrices with the aforementioned minor violations because the corresponding eigenvalues of the produced matrix are not exactly non-negative but effectively 0.} That is, set 
    $\Lambda_{1,j} = \max\{\Lambda_{0,j},0\}$ for each $j = 1,...,p$ and
    $E_{1} = E_{0}$.
    \item Using $\Lambda_{1}$ and $E_{1}$, recompose the correlation matrix $\Omega_{1} = E_{1}\Lambda_{1}E_{1}$.
    \item Compute $\Omega_{2} = L \Omega_{1} L$, where $L_{j,j} = 1/\sqrt{\Omega_{1,{j,j}}}$. 
    \item With $\Omega_{2}$ as an initial value, run SDP again. 
\end{enumerate}

The above steps address the core numerical optimization issue: the initial correlation matrix output is not exactly positive semidefinite because it contains negative eigenvalues. Since positive semidefiniteness requires all eigenvalues to be non-negative, these negative eigenvalues can be rounded to zero, as they are effectively negligible. Such rounding does not affect the standard errors implied by the correlation matrix. Steps 3, 4, and 5 then provide an additional robustness check to ensure that, after decomposition and rounding, the resulting correlation matrix remains valid and optimal.

\subsection{The Karush–Kuhn–Tucker Conditions}\label{sec:kkt}
As the initial outputs of SDP solvers may require adjustment to yield valid correlation matrices, one might worry whether the adjusted correlation matrix obtained via the implementation guidelines in the previous subsection indeed corresponds to the true lower bound of the standard errors.

A natural way to verify the robustness of this approach is to check whether the Karush--Kuhn--Tucker (KKT) conditions are satisfied by the adjusted correlation matrix. The KKT conditions generalize the method of Lagrange multipliers by providing optimality criteria for problems with both equality and inequality constraints. Commonly employed in nonlinear programming, the KKT conditions involve solutions to both the primal problem (the original formulation of the optimization problem) and the dual problem (an alternative but equivalent formulation) to determine whether a global optimum has been attained. When the primal problem is convex, the KKT conditions are both necessary and sufficient for primal and dual optimality, thereby guaranteeing that the solution is globally optimal \citep{boyd}.

By checking the KKT conditions, researchers can confirm that the implied standard errors of the adjusted best-case correlation matrices for particular model parameters are indeed the smallest feasible standard errors. Since most SDP solvers simultaneously compute both primal and dual solutions, verifying the KKT conditions requires only straightforward computations involving these solutions, the known marginal variances of the empirical moments, and the Jacobian of the structural model evaluated at the point estimates of the chosen moments. Proposition \ref{prop_sdp_kkt} below specifies the conditions that can be used to verify the optimality of adjusted correlation matrices.

\begin{prop}
    \label{prop_sdp_kkt}
    Given an adjusted $p \times p$ correlation matrix $\Omega$ with negative eigenvalues rounded to zero and the dual variables for the positive semidefiniteness constraint $\Lambda$, a $p \times p$ matrix, and for the equality constraints, $\Psi \in \mathbb{R}^{p}$, of the SDP \eqref{sdp}--\eqref{sdp_affine_corr}, $\Omega$ solves \eqref{eq:pre_sdp_opt} if and only if $(\Omega,\Lambda)$ satisfies the following KKT conditions: 
    \begin{itemize}
        \item Primal Feasibility: $\Omega \succeq 0$ and $\Omega_{j,j} = 1$ for each $j = 1,...,p$.
        \item Dual Feasibility: $\Lambda \succeq 0$. 
        \item Complementary Slackness: $\Lambda\Omega = \mathbf{0}_{p \times p}$.
        \item Stationarity: $D\hat{\ell}\hat{\ell}'D + \Lambda - \mathrm{diag}(\Psi) = \mathbf{0}_{p \times p}$.
    \end{itemize}
\end{prop}

\noindent
Proof of this statement is found in Appendix \ref{sec:appendix1.2}. 

The \textit{Primal Feasibility} condition requires that the adjusted correlation matrix $\Omega$ satisfies the generalized inequality and equality constraints of the original SDP, as specified in \eqref{sdp_constraint} and \eqref{sdp_affine_corr}. The \textit{Dual Feasibility} condition imposes that the dual variable associated with the positive semidefiniteness constraint is itself positive semidefinite, ensuring that the resulting matrix $\Omega$ is indeed positive semidefinite. The \textit{Complementary Slackness} condition requires that, for each constraint, exactly one of the following holds: (i) the constraint binds with equality, or (ii) the corresponding Lagrange multiplier equals zero. Finally, the \textit{Stationarity} condition, which is analogous to a first-order condition in unconstrained optimization, requires that the gradient of the Lagrangian with respect to $\Omega$ vanishes at the optimal $\Omega$.

\section{Empirical Applications}\label{sec:empirics}
In this section, we apply the SDP method described in Section~\ref{sec:implementation} to compute the best-case standard errors of three empirical applications: two calibrated macro models and an IV study of public-housing effects. We verify these estimates as optimal using the KKT conditions in Proposition \ref{prop_sdp_kkt}, and we compare the best-case standard errors to the worst-case standard errors.

\subsection{Menu Cost Price-Setting in Multiproduct Firm Model}\label{sec:empirics:menu}

The first application is based on \citet{alvarez&lippi:2014} and concerns the model of a multiproduct firm facing a price-setting decision in the context of sticky prices. Price rigidity stems from the need to pay a fixed ``menu'' cost prior to price adjustment. While the firm's desired prices change throughout time, it must weigh the benefits of changing all the prices of its products simultaneously against the loss it will incur in paying the menu cost. The $k = 3$ structural parameters of the model are the number of products of the firm ($m$), the volatility of the firm's desired prices ($\nu$), and the scaled menu cost relative to the curvature of the firm's profit function ($\sqrt{\psi/B}$).\footnote{$\psi$ denotes the menu cost while $B$ represents the curvature of the profit function.} These parameters are calibrated using $p = 4$ empirical moments consisting of the frequency of price changes, and the first, second, and fourth moments of price changes. 

\subsubsection*{Model Overview}
A firm sells $m$ products and sets the prices of each product simultaneously. The set prices remain fixed until the firm pays the menu cost $\psi$, and the firm's desired log prices change in continuous time following the process of $m$ independent random walks without drift and with volatility $\nu$. The price gap is defined as the difference between the firm's desired prices for their $m$ products and the actual prices of the products. Each period, the firm's profits are proportional to the sum of the squared price gaps. $B$ is a proportionality constant representing the loss of charging a price different from the desired level. Thus, $B$ measures the curvature of the profit function. The firm must then minimize the expected discounted cost of selling at prices that deviate from the desired levels and the loss from paying the fixed menu cost. A nonlinear transformation of the structural model parameters imply the observed empirical moments of the number of price adjustments per unit of time ($N_{a}$), the first moment of the absolute log price changes ($\mathbb{E}[|\Delta p_{i}|]$), the second moment ($\mathbb{E}[(\Delta p_{i})^{2}]$), and the fourth moment ($\mathbb{E}[(\Delta p_{i})^{4}]$). Below, the implied moment-matching relationship is rederived using Propositions 4 and 6, and Equation 10 of \citet{alvarez&lippi:2014}:

\begin{align}
    \label{eq:13}
    \begin{bmatrix}
            N_{a} \\
            \mathbb{E}[|\Delta p_{i}|] \\
            \mathbb{E}[(\Delta p_{i})^{2}] \\
            \mathbb{E}[(\Delta p_{i})^{4}]
        \end{bmatrix} = \begin{bmatrix}
            \frac{\nu \sqrt{\frac{m^{2}}{2m + 4}}}{\sqrt{\frac{\psi}{B}}}\\
            \frac{\sqrt{\bar{y}}}{\frac{m-1}{2}\text{Beta}\left(\frac{m-1}{2},\frac{1}{2}\right)}\\
            \frac{\nu\sqrt{2m + 4}}{m\sqrt{\frac{\psi}{B}}}\left(1 + \frac{2}{\pi}\left(1 + \frac{1}{2m}\right)\right)\\
            \frac{3m}{m + 2}\left[\frac{\nu\sqrt{2m + 4}}{m\sqrt{\frac{\psi}{B}}}\left(1 + \frac{2}{\pi}\left(1 + \frac{1}{2m}\right)\right)\right]^{2}
        \end{bmatrix}
\end{align}

\subsubsection*{Empirical Moments}

Following the data cleaning procedure of \citet{alvarezetal:2016}, we fit the calibrated model in Equation \hyperref[eq:13]{\eqref{eq:13}} using scanner data from a single store of the supermarket chain Dominick's for the years 1989 to 1994. For this application, the four empirical moments are computed with the movement data set for beer products.\footnote{This data set, entitled wber.csv, can be found on the Chicago Booth website (\url{https://www.chicagobooth.edu/research/kilts/datasets/dominicks}).} As in \citet{alvarezetal:2016}, we follow a four-step process to clean the data prior to estimation: (1) We drop the data on all stores except for store 122; (2) Any observations that correspond to prices below 20 cents or above 25 dollars are discarded; (3) Observations with absolute price changes less than one percent are set to zero; (4) The largest 1\% of absolute log price changes are also discarded. We use only data on beer products to simplify the inference exercise and make the results of the application easier to interpret. 

This process yields a data set of sample size $N = 37,916$ observations on weekly prices of 499 beer products. On average, each beer product has observations of around 76 weeks. Price changes are treated as i.i.d. processes across beer products and time in the calculation of standard errors. 

For the standard error computations, we use the four estimated empirical moments shown in Table 3 of \citet{mikkel}. We first perform full-information inference (i.e. using the readily estimable correlation structure computed by accessing the full data from a single source). Then, we find the worst-case and best-case standard errors simulating the scenario where we only have the diagonal elements or marginal variances of the estimated covariance matrix. Additionally, we calculate estimates under the erroneous yet commonly used in practice assumption that the empirical moments are mutually independent of one another (i.e. the off-diagonal elements of the covariance matrix are set to zero).

\subsubsection*{Results}

We estimate the structural model parameters using all four moments. The full-information estimates were computed using conventional full-information one-step estimation, and the limited-information best-case and worst-case estimates are calculated with the same
efficient weighting matrix suggested by the Section 3.2 algorithm of \citet{mikkel}. The estimates and standard errors implied by a procedure assuming mutually independent moments rely on an identity weighting matrix. The point estimates\footnote{The point estimates differ slightly across the full-information, limited-information, and independent methods since we use different weighting matrices for each approach. These minor discrepancies are discussed in more detail in Footnote 18 of \citet{mikkel}.} for the number of products, $\hat{m}$, the volatility of the desired prices of the multiproduct firm, $\hat{\nu}$, and the scaled menu cost with respect to the curvature of the firm's profit function $\widehat{\sqrt{\psi/B}}$ and their corresponding standard errors are reported for each framework in \hyperref[table:menucost]{Table \ref{table:menucost}}. 

The Full Information, Independent, and Worst-case $\bar{\sigma}_{\theta}^{CP}/\sqrt{n}$ entries of \hyperref[table:menucost]{Table \ref{table:menucost}} replicate the results of \citet{mikkel}. Our study's main contribution is the computation of the standard errors under the best-case framework, $\underline{\sigma}_{\theta}/\sqrt{n}$, and under the best-case framework with additional information on the correlation structure, $\underline{\sigma}_{\theta}^{\rho_{23}}/\sqrt{n}$. For the latter column, we invoke Proposition \ref{prop_sdp_solution} and make use of the fact that we can compute the correlation between the second and third moments of the data. This allows us to incorporate an additional constraint into the SDP and leverage information that tightens our computation of the best-case standard errors.

The best-case standard errors are substantially smaller than the worst-case standard errors. While the worst-case standard errors are already small, implying that the point estimates of the \citet{alvarez&lippi:2014} structural model parameters are statistically significant, we find that the best-case standard errors are notably smaller and effectively zero. However, when we make use of the estimable correlation between the third and fourth moments\footnote{In this setting, as we have access to the full data, we could have used the correlations between other pairings of moments. All other pairings gave similar results with the lower bounds of the standard errors for $\hat{\nu}$ and $\widehat{\sqrt{\psi/B}}$ not differing significantly from the baseline case. But, when using information on the correlation between the second and fourth moments, the best-case standard error for $\hat{m}$ takes on the value of 0.00394. On the other hand, utilizing the correlation between the third and fourth moments yields a lower bound of 0.0182 for the standard error of $\hat{m}$.}, the best-case standard errors for the parameter $\hat{m}$ increase to about 0.8\% of the point estimate from being effectively zero. 

\begin{table}[tb!]
    \centering
    \caption{Menu Cost: Parameter Estimates and Standard Error Computation}\vspace*{.2cm}
  \resizebox{\textwidth}{!}{
    \begin{tabular}{ccccccccc} \toprule
    \multirow{2}{*}{Parameter} & \multicolumn{2}{c}{Full Information} & \multicolumn{2}{c}{Independent} & \multicolumn{4}{c}{Bounds} \\
    \cmidrule(lr){2-3} \cmidrule(lr){4-5} \cmidrule(lr){6-9}
    & $\hat{\theta}$ & $\hat{\sigma}_{\theta}/\sqrt{n}$ & $\tilde{\theta}$ & $\hat{\sigma}_{\theta}/\sqrt{n}$ & $\check{\theta}$ & $\underline{\sigma}_{\theta}/\sqrt{n}$ & $\underline{\sigma}_{\theta}^{\rho_{23}}/\sqrt{n}$ & $\overline{\sigma}_{\theta}^{CP}/\sqrt{n}$ \\ \midrule
    number of products & 3.255 & 0.051 & 2.829 & 0.091 & 2.786 & 0.0000201 & 0.0221 & 0.148 \\
    Volatility & 0.089 & 0.001 & 0.090 & 0.000 & 0.090 & 0.000429 & 0.000453 & 0.001 \\
    Menu Cost & 0.305 & 0.003 & 0.280 & 0.006 & 0.278 & 0.000 & 0.000 & 0.011 \\ \bottomrule
    \end{tabular}}
    \label{table:menucost}\vspace{.2cm}
    \begin{minipage}{1\textwidth} 
        {\small Note: Full Information estimates use the complete covariance matrix $\hat{\Sigma}$; Independent estimates use only the diagonal elements of $\hat{\Sigma}$. 
        Bounds show best-case ($\underline{\sigma}_{\theta}$ and $\underline{\sigma}_{\theta}^{\rho_{23}}$) standard errors computed via SDP and worst-case ($\overline{\sigma}_{\theta}^{CP}$) standard errors following \cite{mikkel}.  The conventional weighting matrix is used for the full-information estimates, an identity weighting matrix is used for the independent estimates, and the efficient weighting matrix as in the Section 3.2 algorithm of \cite{mikkel} is used for the best-case and worst-case estimates. Parameters reported are the number of products ($\hat{m}$), Volatility ($\hat{\nu}$), and scaled menu cost ($\widehat{\sqrt{\psi/B}}$). Columns under a ``$\theta$'' denote the point estimates. \par}
    \end{minipage}
\end{table}

A few interesting comparisons can also be made with the other estimation frameworks. The next smallest set of standard errors can be found under the full information approach. The first row of \hyperref[table:menucost]{Table \ref{table:menucost}} shows that using the full underlying correlation structure for minimum distance inference leads to standard errors that are much closer to the best-case standard errors. 
This tells us that, for the number of products and Menu Cost parameters, the best-case standard errors consist of over half of the total weight of the full-information standard errors. Moreover, on average, the distance between the full information standard errors and the best-case standard errors is smaller than the distance between the full information standard errors and the worst-case standard errors by approximately 0.02. 

\subsection{Heterogeneous-Agent New Keynesian Model}

For the next application, we perform inference on the general equilibrium HANK macro model of \citet{auclertetal:2021} and \citet{mckayetal:2016} using impulse response function estimates from two separate studies: \citet{changetal:2014} and \citet{agrippino&ricco:2021}. The latter two papers provide independent empirical micro cross-sectional and macro time series moments. Consequently, the covariances among the moments computed from these two papers is not immediately available for full information minimum distance estimation. 

The aim of this application is to study the impacts of productivity and Monetary Policy shocks on aggregate outcomes. The HANK model relates the interactions among a continuum of heterogeneous households, monopolistically competitive firms facing a price adjustment cost, a government that gives out lump sum payments, and a central bank that sets the nominal interest rate. To characterize these interactions, this HANK inference setting concentrates on estimating $k = 7$ structural model parameters: the Taylor Rule coefficient on inflation that determines the central bank's nominal interest rate policy, the slope of the New Keynesian Phillips Curve, the two autoregressive coefficients of the AR(2) process for TFP growth, the standard deviation coefficient for the TFP growth process, and the two autoregressive coefficients of the AR(2) process for the Monetary Policy shock. In this application, $p = 23$ empirical moments are taken from \citet{changetal:2014} and \cite{agrippino&ricco:2021} for the calibration procedure. 

\subsubsection*{Model Overview}
We follow a modified version of the one-asset HANK model discussed in Appendix B.2 of \citet{auclertetal:2021}. In this model, the economy consists of heterogeneous households that choose the amount they consume, the quantity they save in a nominal Treasury bond, and the hours they work. They receive lump sum payments from the government and profits from firms that they own. The firms are monopolistically competitive firms that set prices according to a quadratic adjustment cost, implying a New Keynesian Phillips Curve. The government issues one-period nominal bonds and sets taxes to balance its periodic budget. Following a standard Taylor Rule, the central bank determines the nominal interest rate on bonds. Markets clear when the total goods produced by the firms equals private consumption, public consumption, and the price adjustment costs; labor demand equals labor supply in efficiency units; and aggregate household savings equals government bonds.

To solve this model, we use the first-order linearization method of \citet{auclertetal:2021}. We only estimate the structural model parameters that do not determine the steady state of the model. While it is feasible to estimate parameters that influence the steady-state, this simplification allows us to avoid recomputing the model, evading computational issues. Hence, identical to the approach taken by \citet{mikkel}, we fix the steady-state parameters at the values in Table B.2 of \citet{auclertetal:2021}.

\subsubsection*{Empirical Moments}
\citet{changetal:2014} and \citet{agrippino&ricco:2021} provide two Structural Vector Autoregressions with point estimates and standard errors of impulse responses with respect to the identified shocks. These estimates allow us to study the impact of productivity (TFP) and Monetary Policy shocks. 

From \citet{changetal:2014}, we directly obtain the estimates corresponding to the TFP shock. In particular, we use the reported values of the blue solid lines of Figures 9 and 11 of \citet{changetal:2014} that represent distributional innovations or productivity shocks. The three impulse responses of (1) GDP or aggregate output (in the model), (2) TFP, and (3) the fraction of people with earnings smaller than $2/3$ of GDP per capita are used for this application. The last response variable is taken from cross-sectional data of the Current Population Survey (CPS). Although the CPS has a limited sample size, as noted in \citet{changetal:2014}, the estimation approach for these impulse responses already takes into consideration the uncertainty caused by the small CPS samples. 

On the other hand, from \citet{agrippino&ricco:2021}, the impulse response estimates used for this application are induced by Monetary Policy shocks. We use the reported impulse response estimates from the solid blue lines of Figure 3 of \citet{agrippino&ricco:2021} for the responses of aggregate output (referred to as ``industrial production'' in the original paper), price level (the ``consumer price index''), and the model's annualized nominal interest rate (``1-year Treasury rate''). To match the quarterly structural model, we use the end-of-quarter impulse response estimates. 

For the moment-matching procedure, we concentrate on a few impulse horizons and follow \citet{mikkel}. First, we only study four impulse response horizons: the impact horizon, the 1-quarter horizon, the 2-quarter horizon, and the 8-quarter horizon.  Second, we account for the fact that \citet{agrippino&ricco:2021} impulse responses are normalized such that an impact or shock will generate a 100 basis point increase in the Treasury rate, while \citet{changetal:2014} impulse responses are with respect to a shock of one standard deviation. Finally, following the Bernstein-von Mises theorem, we interpret the reported point estimates of the paper as posterior medians and the standard errors as ones implied by credible intervals with a normal approximation. This is done because both \citet{changetal:2014} and \citet{agrippino&ricco:2021} report Bayesian posterior quantiles. Accounting for all this, our application makes use of $p = 23$ empirical moments.

\subsubsection*{Results}

We estimate the $k = 7$ structural model parameters with the $p = 23$ moments using a diagonal weighting matrix of the inverse marginal variances of each moment: $\hat{W}_{j,j} = 1/\hat{\sigma}_{j}^{2}$. Each minimum distance estimate reported in Table~\ref{table:hank} is an overall optimum resulting from gradient-based numerical optimization with 100 different starting values. Details of the procedure in generating the starting values can be found in Footnote 27 of \citet{mikkel}. 

The estimation approach is typical in the HANK setting since models of this type often imply nonlinear and non-convex objective functions for the minimum distance framework. These objective functions are normally difficult to manipulate analytically (leading to the need for numerical approximation). Moreover, the sometimes non-convex nature of these functions can lead to estimates that only reflect local optima rather than global optima.

\begin{table}[tb!]
    \centering
    \caption{HANK: Parameter Estimates and Standard Error Bounds}\vspace*{.2cm}
    \begin{tabular}{cccc} \toprule
        & $\hat{\theta}$ & $\underline{\sigma}_{\theta}/\sqrt{n}$ & $\overline{\sigma}_{\theta}^{CP}/\sqrt{n}$ \\ \midrule
        Taylor Rule & 1.0600 & 0.0000 & 0.1000 \\
        NKPC & 0.0090 & 0.0024 & 0.0080 \\
        TFP Shock (AR1) & 0.0080 & 0.0044 & 0.1460 \\
        TFP Shock (AR2) & -0.0400 & 0.0000 & 0.1890 \\
        TFP Shock (Std) & 0.0060 & 0.0004 & 0.0005 \\
        Monetary Shock (AR1) & 0.7020 & 0.0073 & 0.1100 \\
        Monetary Shock (AR2) & 0.0750 & 0.0000 & 0.1630 \\ \bottomrule
    \end{tabular}
    \label{table:hank}\vspace{.2cm}
    \begin{minipage}{1\textwidth} 
        {\small Notes: The table shows the structural parameter estimates and their corresponding standard error bounds using a diagonal weighting matrix. This estimation procedure relies on point estimates generated by gradient-based numerical optimization with 100 distinct starting values. The overall optimum is used for marginal standard error computation. See Footnote 27 of \cite{mikkel} for details. 
        Best-case standard errors ($\underline{\sigma}_{\theta}$) are obtained via the SDP approach, while worst-case bounds ($\overline{\sigma}_{\theta}^{CP}$) follow the \cite{mikkel} method. 
        Parameters reported are the Taylor Rule coefficient on inflation (``Taylor Rule''), slope of the New Keynesian Phillips Curve (``NKPC''), first and second autoregressive coefficients (``AR1'', ``AR2''), and standard deviation (``Std'') of the TFP and monetary shock processes.\par}${}$\\
    \end{minipage}
\end{table}

In Table~\ref{table:hank}, under the worst-case standard errors, the two autoregressive coefficients for the TFP shock process and the second autoregressive coefficient for the monetary shock process are implied to be statistically insignificant. However, the gap between the best- and worst-case standard errors for the three parameters is rather large as the best-case standard errors are either effectively 0 or a few decimal places smaller in magnitude compared to the worst-case standard errors. Although the first autoregressive parameter for the monetary shock is unambiguously positive in the estimate of the upper bound, the worst-case approach is uninformative about the impact of TFP impulse responses, an equally vital component of the model. The best-case standard errors then offer an alternative perspective on the TFP parameter estimates and highlight how the conjunctive use of the best- and worst-case standard errors can be of value to researchers. 

In this application, the best-case and worst-case standard errors disagree by a substantial amount. As argued in Section~\ref{sec:intro}, this signals that it may be worthwhile to estimate the underlying correlation structure among the moments used in \citet{changetal:2014} and \citet{agrippino&ricco:2021}. Given the large gap between the best- and worst-case standard errors, three out of the seven structural parameters of the HANK model have confidence intervals that may or may not contain zero. As there is a lot of uncertainty over the estimates and interpretations of the results of the minimum distance approach, exam
ining the correlation structure of the empirical moments becomes a reasonable next step. 

\subsection{The Effect of Public Housing on Children's Outcomes}\label{sec:empirics:ts2sls}
We now apply our method to the two-sample instrumental variables (TSIV) setting originally proposed by \cite{klevmarken_missing_1982} and popularized by \cite{angrist1992effect}. In this framework, the researcher observes two samples that do not contain the full set of variables required for standard IV estimation, but share a common set of instruments. For example, one sample contains the outcome and instruments, and the other contains the endogenous regressor and instruments. We focus on the just-identified version of the two-sample 2SLS (TS2SLS) estimator, which is commonly used for its computational simplicity and statistical efficiency (\cite{inoue_two-sample_2010}).
TS2SLS has been widely used in the literature on intergenerational mobility (\citealp{olivetti_name_2015}; \citealp{barone_intergenerational_2021}), in public-policy studies (\citealp{currie_are_2000}; \citealp{deng_early-life_2022}), and in credit-supply shocks in household and firm contexts (\citealp{crossley_house_2024}; \citealp{renkin_credit_2024}). 

\subsubsection*{Model Overview}
Let sample 1 consist of observations $\{(y_{1i}, z_{1i})\}_{i=1}^{n_1}$, and sample 2 consist of $\{(x_{2i}, z_{2i})\}_{i=1}^{n_2}$, where $y_{1i} \in \mathbb{R}$ is the scalar outcome, $x_{2i} \in \mathbb{R}^k$ is the endogenous regressor, and $z_{1i}, z_{2i} \in \mathbb{R}^q$ are instruments. Let $Y_1$, $X_2$, $Z_1$, and $Z_2$ denote the matrices stacking these vectors across $i$. 
For simplicity, consider a just-identified case where $q=k$. 
The TS2SLS estimator is then given by:
\begin{align*}
    \hat{\theta} = (Z_2'X_2/n_2)^{-1} C Z_1'Y_1/n_1, \quad\text{where}\:  C = (Z_2'Z_2/n_2)^{-1} Z_1'Z_1/n_1.
\end{align*}

This expression shows that $\hat{\theta}$ is a smooth function of two empirical moments
\begin{align*}
    \hat{\mu}_1 = Z_2'X_2/n_2, \quad\text{and}\quad
    \hat{\mu}_2 = Z_1'Y_1/n_1.
\end{align*}

Hence, we can write $\hat{\theta} = h(\hat{\mu}_1,\hat{\mu}_2)$ for a differentiable map $h$, and apply the Delta method to obtain the inference for $\hat{\theta}$.
Under the usual two-sample assumption, the two moment vectors are treated as independent; the covariance block linking $\hat{\mu}_1$ and $\hat{\mu}_2$ is set to zero. This restriction is stronger than requiring the two datasets to contain different individuals: it requires the sampling errors in the estimated first-stage and reduced-form moments to be uncorrelated. Our method relaxes that restriction and delivers sharp lower and upper bounds for the asymptotic variance of $\hat{\theta}$.
In the context of \cite{currie_are_2000}, this relaxation is empirically relevant. Although the CPS and the 1990 Census do not share individual observations, they cover overlapping cohorts and labor markets and are subject to common aggregate conditions and policy environments. These shared features plausibly induce correlation in sampling error across the estimated first-stage and reduced-form moments.

\subsubsection*{Empirical Moments}
We illustrate this using the influential empirical application in \cite{currie_are_2000}, which studies the effect of public housing on children’s outcomes. 
The endogenous variable is whether a household lives in public housing, and the instrument is an indicator of whether the household has children of different genders. This instrument exploits eligibility rules from the Department of Housing and Urban Development, under which families with mixed-gender children are eligible for larger housing units and thus more likely to enter public housing.
The outcome variables are drawn from the 1990–1995 CPS March Supplement (sample 2), while the household covariates come from the 1990 Census (sample 1).\footnote{%
We use the cleaned CPS and Census extracts provided in the replication package of \cite{choi2018weak}.}
The analysis focuses on four outcomes, as in Table 4 of \cite{currie_are_2000}: monthly rent payments, overcrowding (fewer than three rooms), residence in large buildings with more than 50 units, and whether any child has been held back in school.

\subsubsection*{Results}
Estimation results are reported in Table \ref{tab:currieyelowitz}. The first column reproduces the point estimates in \cite{currie_are_2000}. The next two columns report conventional standard errors from \cite{currie_are_2000} and \cite{inoue_two-sample_2010}, respectively.
The remaining columns present explicit sharp bounds on the standard errors derived in Section \ref{sec:analytic} under diagonal and block-diagonal dependence structures. 
Since the block-diagonal covariance structure is fully known and the analytic formulation is feasible in this setting, there is no need to solve the equivalent semidefinite program numerically.
Accordingly, we compute the bounds directly using the closed-form expressions in Section \ref{sec:analytic}. The upper bound for standard errors in each case coincides with the standard error of the worst-case proposed in \cite{mikkel}.

\begin{table}[tb!]
    \centering
    \caption{The Effect of Public Housing on Children's Outcomess: Parameter Estimates and Standard Error Bounds}\vspace*{.2cm}
\resizebox{\textwidth}{!}{%
\begin{tabular}{lccccccc}
\toprule
& & \multicolumn{2}{c}{Conventional SE} & \multicolumn{2}{c}{{Diag only}} & \multicolumn{2}{c}{{Block diag}} \\
\cmidrule(lr){3-4} \cmidrule(lr){5-6} \cmidrule(lr){7-8}
& $\hat{\theta}$ &$\hat{\sigma}_{\theta}^{CY}/\sqrt{n}$& $\hat{\sigma}_{\theta}^{IS}/\sqrt{n}$ & $\underline{\sigma}_{\theta}/\sqrt{n}$ & $\overline{\sigma}_{\theta}^{CP}/\sqrt{n}$ & $\underline{\sigma}_{\theta}/\sqrt{n}$ & $\overline{\sigma}_{\theta}^{CP}/\sqrt{n}$ \\ \midrule
 Rental payment & 0.3717 &0.0589& 0.1125 & 0.0000 & 0.5080 & 0.0369 & 0.1546 \\
 Family is overcrowded & -0.1595 &0.0624& 0.0732 & 0.0000 & 0.3350 & 0.0241 & 0.1007\\ 
 Dense building & -0.1154 &0.0468& 0.0543 & 0.0000 & 0.2479 & 0.0192 & 0.0746 \\
 Child was held back & -0.1113 &0.0691& 0.0666 & 0.0000 & 0.2896 & 0.0389 & 0.0941 \\
 \bottomrule
    \end{tabular}}
    \label{tab:currieyelowitz}\vspace{.2cm}
    \begin{minipage}{1\textwidth} 
        {\small Note: The table shows the parameter estimates, the corresponding standard errors, and their bounds. Standard errors in columns (2)–(3) follow \cite{currie_are_2000} and \citet{inoue_two-sample_2010}.
        Columns (4)–(7) report explicit sharp bounds from Section \ref{sec:analytic}, which coincide with the explicit upper bound in \cite{mikkel}.
        The SDP approach in that paper solves the same problem numerically; here we rely on the analytic closed-form version.
        \par}${}$\\
    \end{minipage}
\end{table}

We first observe that the conventional standard errors lie between our lower and upper bounds.
The diagonal-only case is largely uninformative, with the feasible set spanning from near zero to relatively wide upper bounds.
In contrast, imposing the block-diagonal structure dramatically tightens the range, yielding informative and economically meaningful uncertainty estimates.
These results suggest that accounting for within-block correlation can substantially reduce the identified set for standard errors without relying on independence assumptions.
Because there are only two moments in this application, the block-diagonal constraint is especially informative.


\section{Summary and Discussion}\label{sec:summary}

In this paper, we contribute to the literature on bounding the standard error of parameters estimated from moment conditions across different samples. First, we derive an explicit sharp lower bound in settings where no information about cross-sample correlations is available. Second, we develop computationally feasible sharp bounds based on a semidefinite program, which delivers the best-case standard error in more general settings. Empirical illustrations in both macroeconomic and microeconomic contexts demonstrate the practical usefulness of our approach.
Together with \citet{mikkel}, this paper provides a comprehensive set of tools for computing sharp lower and upper bounds in a wide range of situations. While the existing literature focuses primarily on calibrated parameters in macroeconometric settings, we expand the scope of applications to include two-step estimators that are also relevant to microeconometric settings.

A key takeaway from the empirical applications is how the informativeness of the bounds depends on the structure of the reduced-form moments and on the amount of auxiliary information available. As the analytic expressions in Section \ref{sec:analytic} illustrate, the bounds can become extremely wide when the number of reduced-form moments increases: the lower bound often collapses to zero, while the upper bound can grow very large, yielding limited guidance for inference. However, as shown in Sections \ref{sec:empirics:menu} and \ref{sec:empirics:ts2sls}, incorporating additional information about the correlation matrix can substantially sharpen the bounds, bringing the lower and upper bounds closer and thereby producing a more informative range for the plausible standard errors. In the menu-cost application, even a single strong piece of information leads to a dramatic tightening of the bounds. In the TS2SLS case, with only two reduced-form moments and a block-diagonal restriction on the correlation matrix, the resulting bounds are highly informative, offering practical guidance for empirical researchers while remaining agnostic about arbitrary cross-dataset dependence.

\vspace{1cm}
\bibliographystyle{ecta}
\bibliography{biblio}

@article{vohra2025inference,
  title={Inference for the Marginal Value of Public Funds},
  author={Vohra, Vedant},
  journal={Available at SSRN 5738770},
  year={2025}
}

@article{mikkel,
  title={Standard errors for calibrated parameters},
  author={Cocci, Matthew D and Plagborg-M{\o}ller, Mikkel},
  journal={Review of Economic Studies},
  volume={92},
  number={5},
  pages={2952--2978},
  year={2025},
  publisher={Oxford University Press UK}
}

@incollection{handbook,
	author = {Geert Ridder and Robert Moffitt},
	doi = {https://doi.org/10.1016/S1573-4412(07)06075-8},
	editor = {James J. Heckman and Edward E. Leamer},
	issn = {1573-4412},
	pages = {5469-5547},
	publisher = {Elsevier},
	series = {Handbook of Econometrics},
	title = {The Econometrics of Data Combination},
	url = {https://www.sciencedirect.com/science/article/pii/S1573441207060758},
	volume = {6},
	year = {2007}}

@article{angrist1992effect,
  title={The effect of age at school entry on educational attainment: an application of instrumental variables with moments from two samples},
  author={Angrist, Joshua D and Krueger, Alan B},
  journal={Journal of the American statistical Association},
  volume={87},
  number={418},
  pages={328--336},
  year={1992},
  publisher={Taylor \& Francis}
}

@BOOK{boyd,
  title = {Convex Optimization},
  author = {Stephen Boyd and Lieven Vandenberghe},
  year = {2004},
  publisher = {Cambridge University Press},
}

@article{alvarez&lippi:2014,
    title = {Price Setting with Menu Cost for Multiproduct Firms},
    author = {Fernando Alvarez and Francesco Lippi},
    journal = {Econometrica},
    volume = {82},
    number = {1},
    pages = {89-135},
    year = {2014}
}

@article{alvarezetal:2016,
    title = {The Real Effects of Monetary Shocks in Sticky Price Models: A Sufficient Statistic Approach},
    author = {Fernando Alvarez and Hervé Le Bihan and Francesco Lippi},
    journal = {American Economic Review},
    volume = {106},
    number = {10},
    pages = {2817-2851},
    year = {2016}
}

@article{auclertetal:2021,
    title = {Using the Sequence-Space Jacobian to Solve and Estimate Heterogeneous-Agent Models},
    author = {Adrien Auclert and Bence Bardóczy and Matthew Rognlie and Ludwing Straub},
    journal = {Econometrica},
    volume = {89},
    number = {5},
    pages = {2375-2408},
    year = {2021}
}

@article{mckayetal:2016,
    title = {The Power of Forward Guidance Revisited},
    author = {Alisdair McKay and Emi Nakamura and Jón Steinsson},
    journal = {American Economic Review},
    volume = {106},
    number = {10},
    pages = {3133-3158},
    year = {2016}
}

@article{changetal:2014,
    title = {Heterogeneity and Aggregate Fluctuations},
    author = {Minsu Chang and Xiaohong Chen and Frank Schorfheide},
    journal = {National Bureau of Economic Research Working Paper Series},
    year = {2023}
}

@article{agrippino&ricco:2021,
    title = {The Transmission of Monetary Policy Shocks},
    author = {Silvia Miranda-Agrippino and Giovanni Ricco},
    journal = {American Economic Journal: Macroeconomics},
    volume = {13},
    number = {3},
    pages = {74-107},
    year = {2021}
}

@article{choi2018weak,
  title={Weak-instrument robust inference for two-sample instrumental variables regression},
  author={Choi, Jaerim and Gu, Jiaying and Shen, Shu},
  journal={Journal of Applied Econometrics},
  volume={33},
  number={1},
  pages={109--125},
  year={2018},
  publisher={Wiley Online Library}
}

@article{renkin_credit_2024,
	author = {Renkin, Tobias and Z{\"u}llig, Gabriel},
	doi = {10.1257/mac.20220079},
	issn = {1945-7707},
	journal = {American Economic Journal: Macroeconomics},
	language = {en},
	month = apr,
	number = {2},
	pages = {1--28},
	shorttitle = {Credit {Supply} {Shocks} and {Prices}},
	title = {Credit {Supply} {Shocks} and {Prices}: {Evidence} from {Danish} {Firms}},
	url = {https://www.aeaweb.org/articles?id=10.1257/mac.20220079},
	urldate = {2025-07-29},
	volume = {16},
	year = {2024}}

@article{currie_are_2000,
	author = {Currie, Janet and Yelowitz, Aaron},
	doi = {10.1016/S0047-2727(99)00065-1},
	issn = {0047-2727},
	journal = {Journal of Public Economics},
	month = jan,
	number = {1},
	pages = {99--124},
	title = {Are public housing projects good for kids?},
	url = {https://www.sciencedirect.com/science/article/pii/S0047272799000651},
	urldate = {2025-07-28},
	volume = {75},
	year = {2000}}

@article{inoue_two-sample_2010,
	author = {Inoue, Atsushi and Solon, Gary},
	doi = {10.1162/REST_a_00011},
	issn = {0034-6535},
	journal = {The Review of Economics and Statistics},
	month = aug,
	number = {3},
	pages = {557--561},
	title = {Two-{Sample} {Instrumental} {Variables} {Estimators}},
	url = {https://doi.org/10.1162/REST_a_00011},
	urldate = {2025-07-28},
	volume = {92},
	year = {2010}}

@article{barone_intergenerational_2021,
	author = {Barone, Guglielmo and Mocetti, Sauro},
	doi = {10.1093/restud/rdaa075},
	issn = {0034-6527},
	journal = {The Review of economic studies},
	language = {eng},
	note = {Publisher: Oxford University Press},
	number = {4 (321)},
	pages = {1863--1891},
	shorttitle = {Intergenerational {Mobility} in the {Very} {Long} {Run}},
	title = {Intergenerational {Mobility} in the {Very} {Long} {Run}: {Florence} 1427--2011},
	volume = {88},
	year = {2021}}

@article{olivetti_name_2015,
	author = {Olivetti, Claudia and Paserman, M. Daniele},
	doi = {10.1257/aer.20130821},
	issn = {0002-8282},
	journal = {American Economic Review},
	language = {en},
	month = aug,
	number = {8},
	pages = {2695--2724},
	shorttitle = {In the {Name} of the {Son} (and the {Daughter})},
	title = {In the {Name} of the {Son} (and the {Daughter}): {Intergenerational} {Mobility} in the {United} {States}, 1850-1940},
	url = {https://www.aeaweb.org/articles?id=10.1257/aer.20130821},
	urldate = {2025-07-29},
	volume = {105},
	year = {2015}}

@article{crossley_house_2024,
	author = {Crossley, Thomas F. and Levell, Peter and Low, Hamish},
	doi = {10.1016/j.jebo.2024.05.002},
	issn = {0167-2681},
	journal = {Journal of Economic Behavior \& Organization},
	month = jul,
	pages = {86--105},
	title = {House price rises and borrowing to invest},
	url = {https://www.sciencedirect.com/science/article/pii/S0167268124001707},
	urldate = {2025-07-29},
	volume = {223},
	year = {2024}}

@article{deng_early-life_2022,
	author = {Deng, Zichen and Lindeboom, Maarten},
	copyright = {{\copyright} 2022 The Authors. Journal of Applied Econometrics published by John Wiley \& Sons Ltd.},
	doi = {10.1002/jae.2897},
	issn = {1099-1255},
	journal = {Journal of Applied Econometrics},
	language = {en},
	note = {\_eprint: https://onlinelibrary.wiley.com/doi/pdf/10.1002/jae.2897},
	number = {4},
	pages = {771--787},
	title = {Early-life famine exposure, hunger recall, and later-life health},
	url = {https://onlinelibrary.wiley.com/doi/abs/10.1002/jae.2897},
	urldate = {2025-07-29},
	volume = {37},
	year = {2022}}

@techreport{klevmarken_missing_1982,
	author = {Klevmarken, Anders},
	copyright = {http://www.econstor.eu/dspace/Nutzungsbedingungen},
	institution = {IUI Working Paper},
	language = {eng},
	number = {62},
	title = {Missing {Variables} and {Two}-{Stage} {Least}-{Squares} {Estimation} from {More} than {One} {Data} {Set}},
	type = {Working {Paper}},
	url = {https://www.econstor.eu/handle/10419/95205},
	urldate = {2025-07-28},
	year = {1982}}

@article{kapovich1995moduli,
	author = {Kapovich, Michael and Millson, John},
	journal = {Journal of Differential Geometry},
	number = {2},
	pages = {430--464},
	publisher = {Lehigh University},
	title = {On the moduli space of polygons in the Euclidean plane},
	volume = {42},
	year = {1995}}

@article{pagan1984econometric,
  title={Econometric issues in the analysis of regressions with generated regressors},
  author={Pagan, Adrian},
  journal={International economic review},
  pages={221--247},
  year={1984},
  publisher={JSTOR}
}

@article{kydland&prescott:1982,
    title = {Time To Build and Aggregate Fluctuations},
    author = {Finn E. Kydland and Edward C. Prescott},
    journal = {Econometrica},
    volume = {50},
    number = {6},
    pages = {1345-1370},
    year = {1982}
}

@article{gourinchas&parker:2003,
    title = {Consumption Over the Life Cycle},
    author = {Pierre-Olivier Gourinchas and Jonathan A. Parker},
    journal = {Econometrica},
    volume = {70},
    number = {1},
    pages = {47-89},
    year = {2003}
}

@article{kaplan&violante:2014,
    title = {A Model of Consumption Response to Fiscal Stimulus Payments},
    author = {Greg Kaplan and Giovanni L. Violante},
    journal = {Econometrica},
    volume = {82},
    number = {4},
    pages = {1199-1239},
    year = {2014}
}

@article{arellano:2008,
    title = {Default Risk and Income Fluctuations in Emerging Economies},
    author = {Cristina Arellano},
    journal = {American Economic Review},
    volume = {98},
    number = {3},
    pages = {690-712},
    year = {2008}
}

@article{eatonetal:2011,
    title = {An Anatomy of International Trade: Evidence From French Firms},
    author = {Jonathan Eaton and Samuel Kortum and Francis Kramarz},
    journal = {Econometrica},
    volume = {79},
    number = {5},
    pages = {1453-1498},
    year = {2011}
}

@article{nakamura&steinsson:2018,
    title = {Identification in Macroeconomics},
    author = {Emi Nakamura and Jón Steinsson},
    journal = {Journal of Economic Perspectives},
    volume = {32},
    number = {3},
    pages = {59-86},
    year = {2018}
}

\vspace{1cm}
\appendix
\section*{Appendix}
\section{Proofs}

This section collects mathematical proofs of the main results.

\subsection{Proof of Theorem \ref{thm:explicit}}\label{sec:thm:explicit}

\begin{proof}
Let $\Xi = (\ell\ell') \circ \Sigma$, where $\circ$ denotes the Hadamard product, i.e., its entry in the $i$-th row and $j$-th column is $\Xi_{ij} = \ell_i \Sigma_{ij} \ell_j$ for $i,j \in \{1,\cdots,p\}$.
Let $v \sim N(0,\Xi)$.
By \eqref{eq:phi_linear} and \eqref{eq:mu_asym},
we have
\begin{align*}
    \sigma = \left\Vert \sum_{j=1}^{p}v_j\right\Vert_2,
\end{align*}
where $\left\Vert \cdot \right\Vert_2$ denotes the $L^2$-norm defined by $\left\Vert u \right\Vert_2 = \left( E u^2 \right)^{1/2}$ for any scalar random variable $u$.
Therefore, it suffices to find the sharp lower and upper bounds of $\left( E u^2 \right)^{1/2}$ under the knowledge of only the diagonal elements $(\Xi_{11},\cdots,\Xi_{pp})' = (\ell_1^2 \Sigma_{11},\cdots,\ell_p^2 \Sigma_{pp})'$ of $\Xi$.

The upper bound follows from the triangle inequality
$$
\left\Vert\sum_{j=1}^p v_j\right\Vert_2 \le \sum_{j=1}^p \Xi_{jj}^{1/2}.
$$
This upper bound is sharp because the random vector $v = (\Xi_{11}^{1/2},\cdots,\Xi_{pp}^{1/2})' Z$, where $Z \sim N(0,1)$, attains this upper bound as
$
\left\Vert\sum_{j=1}^p v_j\right\Vert_2 = \sum_{j=1}^p \Xi_{jj}^{1/2}.
$

The lower bound also follows from the triangle inequality
\begin{align*}
\left\Vert \sum_{j=1}^p v_j \right\Vert_2
= 
\left\Vert v_m + \sum_{j\neq m} v_j \right\Vert_2
\ge 
\left\Vert v_m\right\Vert_2 - \sum_{i\neq m} \left\Vert v_i \right\Vert_2 
\ge 
\Xi_{mm}^{1/2} - \sum_{j\neq m} \Xi_{jj}^{1/2}.
\end{align*}
Since a norm is non-negative, the truncated lower bound is 
\begin{align*}
\left\Vert \sum_{j=1}^p v_j \right\Vert_2 
\ge 
\max\left\{ \Xi_{mm}^{1/2} - \sum_{j\neq m} \Xi_{jj}^{1/2}, 0\right\}.
\end{align*}
To show the sharpness of this lower bound, we branch into two cases depending on whether
$
\Xi_{mm}^{1/2} > \sum_{j\neq m} \Xi_{jj}^{1/2}
$
or not.

First, consider the case of 
$
\Xi_{mm}^{1/2} > \sum_{j\neq m} \Xi_{jj}^{1/2}.
$
In this case, define the random vector $v$ such that 
$$
v_m = \Xi_{mm}^{1/2} Z 
\qquad\text{ and }\qquad 
v_j = -\Xi_{jj}^{1/2} Z
\text{ for all } j \neq m,
$$
where $Z \sim N(0,1)$.
Then, the norm of the sum of this random vector attains the lower bound
$
\left\Vert \sum_{j=1}^p v_j \right\Vert_2
=
\Xi_{mm}^{1/2} - \sum_{j\neq m} \Xi_{jj}^{1/2}
$
by construction.

Second, consider the case of 
$
\Xi_{mm}^{1/2} \leq \sum_{j\neq m} \Xi_{jj}^{1/2}.
$
In this case, there exists a vector 
$(\theta_1,\cdots,\theta_p)' \in \mathbb{R}^p$ of $p$ angles such that $\sum_{j=1}^p \Xi_{jj}^{1/2} (\cos \theta_j, \sin \theta_j) = (0,0)$ by the polygon inequality \citep[][Lemma 1]{kapovich1995moduli}.
Let $Z,W \sim N(0,1)$ be two independent random variables, and construct a random vector $v$ by setting
$$
v_j
=
\Xi_{jj}^{1/2} \left(\cos\theta_j \cdot Z + \sin\theta_j \cdot W\right)
\qquad\text{ for each } j \in \{1,\cdots,p\}
$$
with such an angle vector $(\theta_1,\cdots,\theta_p)'$.
Then, it follows that 
$$
\sum_{j=1}^p v_j = \left(\sum_{j=1}^p \Xi_{jj}^{1/2} \cos \theta_j\right)Z + \left(\sum_{j=1}^p \Xi_{jj}^{1/2} \sin \theta_j\right)W = 0,
$$
and the norm of the sum of this degenerate random vector $v$ thus attains the lower bound
$
\left\Vert \sum_{j=1}^p v_j \right\Vert_2
=
0.
$
\end{proof}

\subsection{Proof of Proposition \ref{prop:consistency}}\label{sec:prop:consistency}
\begin{proof}
By Assumption \ref{asm:sufficient2} (i), we can decompose $\mathcal{I} = \mathcal{I}_D \cup \mathcal{I}_O$, where $\mathcal{I}_D = \{(1,1),...,(p,p)\}$ consists of all the diagonal index pairs and $\mathcal{I}_O = \mathcal{I} \backslash \mathcal{I}_D$ consists of off-diagonal index pairs from $\mathcal{I}$, if any exist.
Define
\begin{align*}
\mathcal{C}^p(\mathcal{I}) \equiv \left\{R \succeq 0 \ : \ R = R', R_{jk}=1 \text{ for all } (j,k) \in \mathcal{I}_D, \text{ and } R_{jk}=0 \text{ for all } (j,k) \in \mathcal{I}_O \right\}.
\end{align*}
This set $\mathcal{C}^p(\mathcal{I})$ is bounded because $(e_j \pm e_k)' R (e_j \pm e_k) \geq 0$ implies $R_{jk} \in [-1,1]$.
Furthermore, this set $\mathcal{C}^p(\mathcal{I})$ is also closed because the positive semidefinite cone and the affine constraint set (characterized by $R=R'$, $R_{jk}=1$ for all $(j,k) \in \mathcal{I}_D$, and $R_{jk}=0$ for all $(j,k) \in \mathcal{I}_O$) are both closed.

For any $\Tilde\Sigma \in S(\Sigma,\mathcal{I})$, we can write
\begin{align*}
\ell'\Tilde\Sigma\ell
=&
l'Rl + \sum_{(j,k) \in \mathcal{I}_O} \ell_j \Sigma_{jk} \ell_k,
\end{align*}
where $l \equiv (\ell_1 s_1, ...,\ell_p s_p)'$ and $R \in \mathcal{C}^p(\mathcal{I})$.
Likewise, for any $\Tilde\Sigma \in S(\hat\Sigma,\mathcal{I})$, we can write
\begin{align*}
\hat\ell'\Tilde\Sigma\hat\ell
=&
\hat l' R \hat l
+ \sum_{(j,k) \in \mathcal{I}_O} \hat\ell_j \hat\Sigma_{jk} \hat\ell_k,
\end{align*}
where $\hat l \equiv (\hat \ell_1 \hat s_1, ...,\hat \ell_p \hat s_p)'$ and $R \in \mathcal{C}^p(\mathcal{I})$.

Observe that 
$
\hat\ell_i \hat\Sigma_{ij} \hat\ell_j
\stackrel{p}{\to}
\ell_i \Sigma_{ij} \ell_j
$ 
for each $(i,j) \in \mathcal{I}_O$ by the Continuous Mapping Theorem under Assumption \ref{asm:sufficient2} (ii).
Thus, it remains to show
\begin{align*}
    \varphi(\hat l) \stackrel{p}{\rightarrow} \varphi(l)
    \qquad\text{where } \varphi(l) \equiv \min_{R \in \mathcal{C}^p(\mathcal{I})} l' R l.
\end{align*}
Note that $\mathcal{C}^p(\mathcal{I})$ is compact as it is closed and bounded as argued earlier.
Also, note that the mapping $(l,R) \mapsto l'R'$ is continuous.
Therefore, by Berge's Maximum Theorem, $\varphi: \mathbb{R}^p \rightarrow \mathbb{R}$ is continuous.
Since $\hat\ell_j \hat s_j \stackrel{p}{\rightarrow} \ell_j s_j$ for each $j = 1,...,p$ by Assumption \ref{asm:sufficient2} (i)--(ii),
it follows from the Continuous Mapping Theorem that $\varphi(\hat l) \stackrel{p}{\rightarrow} \varphi(l)$.

Combining these results together yields
\begin{equation*}
\min_{\Tilde{\Sigma} \in S(\hat\Sigma,\mathcal{I})} {\hat{\ell}'\Tilde{\Sigma}\hat{\ell}}
\ \ \stackrel{p}{\longrightarrow} \ \
\min_{\Tilde{\Sigma} \in S(\Sigma,\mathcal{I})} {{\ell}'\Tilde{\Sigma}{\ell}}.
\end{equation*}
Another application of the Continuous Mapping Theorem proves the proposition.
\end{proof}

\color{black}

\subsection{Proof of Proposition \ref{prop_sdp_solution}}
\label{sec:appendix1.1}

\begin{proof}
Throughout this proof, we omit the hats from our notation.
    First, we shall show that we can rewrite the minimum distance objective function as the objective function in Equation \eqref{sdp}. Note that any covariance matrix $\Sigma$ can be written as a product of two diagonal matrices and a corresponding correlation matrix:
        \begin{align*}
            \Sigma &= D \Omega D\\
            \text{where } D &\equiv \mathrm{diag}(\Sigma)^{\frac{1}{2}}
        \end{align*}
        Since the expression inside the square root is an asymptotic variance, it must be that 
        \begin{align*}
            \ell'\Tilde{\Sigma}\ell &\geq 0\\
            \implies \ell'D\Omega D\ell &\geq 0
        \end{align*}
        Consequently, as $f(x) = \sqrt{x}$ is a monotonic transformation for any $x \geq 0$, the minimal solution to Equation \eqref{eq:pre_sdp_opt} is directly related to the minimal solution in Equation \eqref{sdp}:
        \begin{align*}
            \Sigma_{\mathrm{min}} = D\Omega_{\mathrm{min}}D
        \end{align*}
        Second, note that the constraints of the optimization problem in \eqref{eq:pre_sdp_opt} are equivalent to the constraints outlined in Equations~\eqref{sdp_constraint} and~\eqref{sdp_affine_corr} as this is equivalent to the statement that $\Omega$ is a correlation matrix $\Longleftrightarrow \Omega_{j,j} = 1 \ \forall j$ and $\Omega$ is a positive semidefinite matrix. This follows from the fact that a covariance matrix is always positive semidefinite. 
        
        By invoking the known property that SDP is a type of convex optimization problem (\citet{boyd}), we have that the solution to the SDP in Equations \eqref{sdp}, \eqref{sdp_constraint}, and \eqref{sdp_affine_corr} is a global minimum. Therefore, solving the optimization problem implied by the minimum distance framework is equivalent to solving a simple SDP.
\end{proof}

\subsection{Proof of Proposition \ref{prop_sdp_kkt}}
\label{sec:appendix1.2}
We now derive the KKT conditions for the nonstandard SDP specified in Section~\ref{bcs:derivation}. As the SDP defined in \hyperref[prop_sdp_kkt]{Proposition \ref{prop_sdp_kkt}} is slightly different from the standard-form SDP in most of the literature, re-deriving the KKT conditions is a necessary step to ensure that the correct optimality criteria are being checked for robustness. The proof that the conditions noted in Proposition~\ref{prop_sdp_kkt} are equivalent to the necessary and sufficient optimality conditions for convex programs begins with the claim that allows us to reduce the nonlinear semidefiniteness constraint into a simple linear constraint involving the trace operator for matrices. We shall first formally state and then prove this claim. 

\noindent \textbf{Claim:} \textit{For any symmetric matrix $B$, $\max_{A \succeq 0}-\langle A,B\rangle$ is finite if and only if $B \succeq 0$.}
\begin{proof}
    Note that the inner product between two symmetric matrices $A$ and $B$ is defined as the trace of the product of these matrices: $\langle A,B\rangle \equiv \mathrm{tr}(AB)$. 
    
    For the proof of necessity, note that if $B \succeq 0$, then we immediately have that $\langle A,B\rangle \geq 0$ since the product of two symmetric positive semidefinite matrices is positive semidefinite and hence their trace is non-negative. Given this, it must also be that $\max_{A \succeq 0}-\langle A,B\rangle \leq 0$. From this sharp bound, the maximum value achieved is then zero and consequently finite. 

    For sufficiency, assume for a contradiction that $B \not\succeq 0$, then we can choose an $A \succeq 0$ such that $\max_{A \succeq 0}-\langle A,B\rangle$ becomes as negative as desired. This can be done by selecting an $A$ with large enough positive eigenvalue, which when interacted with a negative eigenvalue of $B \not\succeq 0$, becomes arbitrarily large and negative. Negating $\max_{A \succeq 0}-\langle A,B\rangle$, which recovers the original expression, yields a maximal value that can be made to be arbitrarily large and hence non-finite. 
\end{proof}

This claim allows us to impose the linear Lagrangian constraint of the trace for our dual function in the KKT derivations. With this result, we can proceed to deriving the full KKT conditions for our nonstandard SDP: 

\begin{proof}
        Recall that the original SDP has the following primal problem: 
        \begin{align*}
            \min_{\Omega} \ &\ell D\Omega D\ell'\\
            \text{s.t. } &\Omega \succeq 0\\
            &\Omega_{i,i} = 1 \text{ for } i = 1,...,p
        \end{align*}
        The primal feasibility conditions follow immediately by definition as we simply have to check that the constraints in the above problem are satisfied by any $\Omega$ we want to check for optimality. We can write the Lagrangian of the problem as 
        \begin{align*}
            \mathcal{L}(\Omega, \Lambda,\Psi) = \ell D\Omega D\ell' + \mathrm{tr}(\Lambda^{T}\Omega) + \sum_{i=1}^{p}\Psi_{i}(1-\Omega_{i,i})
        \end{align*}
        The Lagrangian variable for the equality constraint is straightforward.
        
        Then, we take the derivative of the Lagrangian with respect to the $p \times p$ matrix $\Omega$. We shall do this piece-by-piece. For the first term, note that $\ell D \Omega D\ell'$ is a scalar so we can write $\ell D \Omega D\ell' = \mathrm{tr}(\ell D \Omega D\ell')$. Hence, 
        \begin{align*}
            \frac{\partial}{\partial \Omega}\ell D \Omega D\ell' &= \frac{\partial}{\partial \Omega}\mathrm{tr}(\ell D \Omega D\ell')\\
            &= (\ell D)^{T}(D\ell')^{T}\\
            &= D^{T}\ell'\ell D^{T}\\
            &= D\ell'\ell D
        \end{align*}
        where the second line follows the matrix cookbook rule that $\frac{\partial}{\partial X}\mathrm{tr}(AXB) = A^{T}B^{T}$, and the last line follows from the fact that $D$ is a diagonal matrix. 

        For the second term, note that $\frac{\partial}{\partial X}\mathrm{tr}(XA) = A^{T}$ and so 
        \begin{align*}
            \frac{\partial}{\partial \Omega}\mathrm{tr}(\Lambda^{T}\Omega) = \Lambda
        \end{align*}

        Finally, for the last component, we can see that for each term in the sum that 
        \begin{align*}
            \frac{\partial}{\partial \Omega}\Psi_{i}(1-\Omega_{i,i}) = \begin{pmatrix}
                0 & 0 & 0 & \hdots & 0\\
                0 & \ddots & \vdots & \hdots & 0\\
                0 & 0 & \Omega_{i,i} & \hdots & \vdots\\
                \vdots & \vdots & \vdots & \ddots & 0\\
                0 & 0 & 0 & 0 & 0\\
            \end{pmatrix}
        \end{align*}
        and so, the derivative of the last expression is 
        \begin{align*}
            \frac{\partial}{\partial \Omega}\sum_{i=1}^{p}\Psi_{i}(1-\Omega_{i,i}) = \begin{pmatrix}
                -\Psi_{1} & 0 & 0 & \hdots & 0\\
                0 & -\Psi_{2} & 0 & \hdots & 0\\
                \vdots & \vdots & \ddots & \hdots & \vdots\\
                0 & 0 & 0 & \ddots & 0\\
                0 & 0 & 0 & 0 & -\Psi_{p}
            \end{pmatrix} = -\mathrm{diag}(\Psi)
        \end{align*}
        Putting this all together yields 
        \begin{align*}
            \frac{\partial}{\partial \Omega}\mathcal{L}(\Omega, \Lambda, \Psi) = D\ell'\ell D + \Lambda - \mathrm{diag}(\Psi)
        \end{align*}
        This allows us to write the corresponding dual function to our primal as 
        \begin{align*}
            g(\Lambda, \Psi) &= \inf_{\Omega}\mathcal{L}(\Omega,\Lambda,\Psi)\\
            &= \begin{cases}
                -\sum_{i=1}^{p}\Psi_{i}, &\text{if } D\ell'\ell D + \Lambda - \mathrm{diag}(\Psi) = 0\\
                -\infty, &\text{otherwise}
            \end{cases}
        \end{align*}
        Thus, the dual problem is 
        \begin{align*}
            \max \ &g(\Lambda,\Psi) = \begin{cases}
                \sum_{i=1}^{p}\Psi_{i}, &\text{if } D\ell'\ell D + \Lambda - \mathrm{diag}(\Psi) = 0\\
                -\infty, &\text{otherwise}
            \end{cases}\\
            \text{s.t. } &\Lambda \succeq 0
        \end{align*}
        Or more succinctly 
        \begin{align*}
            \max \ &\sum_{i=1}^{p}\Psi_{i}\\
            \text{s.t. } &D\ell'\ell D + \Lambda - \mathrm{diag}(\Psi) = 0\\
            &\Lambda \succeq 0
        \end{align*}
        Using the constraints of this dual problem, we get the dual feasibility constraint from the matrix inequality restriction on $\Lambda$ and the stationarity constraint from the equality restriction. The complementary slackness constraint follows from taking the product of the primal and dual variables restricted by matrix inequalities, namely $\Omega$ and $\Lambda$. 
    \end{proof}

\subsection{Proof of Proposition \ref{prop_se_range}}
\label{proof_prop_range}
\begin{proof}
    For notation, let $\Sigma_{\mathrm{F}}$ be the full-information covariance matrix, $\Sigma_{\mathrm{BCS}}$ be the best-case covariance matrix, and $\Sigma_{\mathrm{WCS}}$ be the worst-case covariance matrix. 

    First, note that by definition, $c\Sigma c^{T} \geq 0, \forall c,\Sigma$ since $\Sigma$ is positive semidefinite and that 
    \begin{align*}
        \ell\Sigma_{\mathrm{BCS}}\ell' &\leq \ell\Sigma_{\mathrm{WCS}} \ell'\\
        \ell\Sigma_{\mathrm{F}}\ell' &\geq \ell\Sigma_{\mathrm{BCS}}\ell'\\
        \ell\Sigma_{\mathrm{F}}\ell' &\leq \ell\Sigma_{\mathrm{WCS}}\ell'
    \end{align*}
    Second, we shall prove that the interval of values between the lower and upper bound of the standard errors $[\ell\Sigma_{\mathrm{BCS}}\ell', \ell\Sigma_{\mathrm{WCS}}\ell']$ is convex. Begin by noting that for $\Sigma_{\mathrm{BCS}} \succeq 0, \Sigma_{\mathrm{WCS}} \succeq 0$, and $\alpha \in [0,1]$, we have
    \begin{align*}
        z^{T}(\alpha \ell\Sigma_{\mathrm{BCS}}\ell' + (1-\alpha)\ell\Sigma_{\mathrm{WCS}}\ell')z = \alpha z^{T}\ell\Sigma_{\mathrm{BCS}}\ell'z + (1-\alpha)z^{T}\ell\Sigma_{\mathrm{WCS}}\ell'z \geq 0
    \end{align*}
    which works for all $z$ by the positive semidefiniteness of $\Sigma_{\mathrm{BCS}}$ and $\Sigma_{\mathrm{WCS}}$. Moreover, since $\ell\Sigma_{\mathrm{BCS}}\ell' \leq \ell\Sigma_{\mathrm{WCS}}\ell'$, we have that 
    \begin{align*}
        &\alpha \ell\Sigma_{\mathrm{BCS}}\ell' + (1-\alpha)\ell\Sigma_{\mathrm{WCS}}\ell' \leq \alpha \ell\Sigma_{\mathrm{WCS}}\ell' + (1-\alpha)\ell\Sigma_{\mathrm{WCS}}\ell' = \ell\Sigma_{\mathrm{WCS}}\ell'\\
        &\alpha \ell\Sigma_{\mathrm{BCS}}\ell' + (1-\alpha)\ell\Sigma_{\mathrm{WCS}}\ell' \geq \alpha \ell\Sigma_{\mathrm{BCS}}\ell' + (1-\alpha)\ell\Sigma_{\mathrm{BCS}}\ell' = \ell\Sigma_{\mathrm{BCS}}\ell'
    \end{align*}
    and so $\forall \alpha \in [0,1], \alpha \ell\Sigma_{\mathrm{BCS}}\ell' + (1-\alpha)\ell\Sigma_{\mathrm{WCS}}\ell' \in [\ell\Sigma_{\mathrm{BCS}}\ell', \ell\Sigma_{\mathrm{WCS}}\ell']$. Moreover, since the interval $[\ell\Sigma_{\mathrm{BCS}}\ell', \ell\Sigma_{\mathrm{WCS}}\ell']$ is convex, it is also connected and contains no gaps. Finally, since $\ell\Sigma_{\mathrm{F}}\ell' \in [\ell\Sigma_{\mathrm{BCS}}\ell', \ell\Sigma_{\mathrm{WCS}}\ell']$, then $\ell\Sigma_{\mathrm{F}}\ell'$ is of the form 
    \begin{align*}
        \ell\Sigma_{\mathrm{F}}\ell' = \alpha \ell\Sigma_{\mathrm{BCS}}\ell' + (1-\alpha)\ell\Sigma_{\mathrm{WCS}}\ell'
    \end{align*}
    for some $\alpha \in [0,1]$.
\end{proof}

\subsection{A two-moment zero-cancellation calculation}
\label{proof_prop_0bcse}
\begin{proof}
     Recall that under the covariance matrix decomposition into the correlation matrix, our asymptotic variance takes the form 
    \begin{align*}
        \ell\Sigma \ell' = \ell D\Omega D\ell'
    \end{align*}
    where $\ell$ is defined as in Equation~\eqref{eq:min_distance_loadings}, $\Omega$ is the correlation matrix of the empirical moments, and $D$ is a diagonal matrix of the marginal variances of the moments. With $\ell_{i}$ defined as the $i$-th entry of $\ell$, deriving the condition in Proposition \ref{prop_se_range} requires setting the above equation to zero and isolating the correlation variable, $\rho$. 
    \begin{align*}
        \begin{bmatrix}
            \ell_{1} & \ell_{2}
        \end{bmatrix}\begin{bmatrix}
            s_{1}^{2} & 0\\
            0 & s_{2}^{2}
        \end{bmatrix}\begin{bmatrix}
            1 & \rho \\
            \rho & 1
        \end{bmatrix}\begin{bmatrix}
            s_{1}^{2} & 0\\
            0 & s_{2}^{2}
        \end{bmatrix}\begin{bmatrix}
            \ell_{1}\\
            \ell_{2}
        \end{bmatrix} &= 0\\
        \begin{bmatrix}
            \ell_{1}s_{1}^{2} & \ell_{2}s_{2}^{2}
        \end{bmatrix}\begin{bmatrix}
            1 & \rho\\
            \rho & 1
        \end{bmatrix}
        \begin{bmatrix}
            s_{1}^{2} & 0\\
            0 & s_{2}^{2}
        \end{bmatrix}\begin{bmatrix}
            \ell_{1}\\
            \ell_{2}
        \end{bmatrix} &= 0\\
        \begin{bmatrix}
            \ell_{1}s_{1}^{2} + \rho \ell_{2}s_{2}^{2} & \rho \ell_{1}s_{1}^{2} + \ell_{2}s_{2}^{2}
        \end{bmatrix}\begin{bmatrix}
            s_{1}^{2} & 0\\
            0 & s_{2}^{2}
        \end{bmatrix}\begin{bmatrix}
            \ell_{1}\\
            \ell_{2}
        \end{bmatrix} &= 0\\
        \begin{bmatrix}
            \ell_{1}s_{1}^{4} + \rho \ell_{2}s_{1}^{2}s_{2}^{2} & \rho \ell_{1}s_{1}^{2}s_{2}^{2} + \ell_{2}s_{2}^{4} 
        \end{bmatrix}\begin{bmatrix}
            \ell_{1}\\
            \ell_{2}
        \end{bmatrix} &= 0\\
        \implies \ell_{1}^{2}s_{1}^{4} + \rho \ell_{1}\ell_{2}s_{1}^{2}s_{2}^{2} + \rho \ell_{1}\ell_{2}s_{1}^{2}s_{2}^{2} + \ell_{2}^{2}s_{2}^{4} &= 0\\
        \implies \rho = -\left(\frac{\ell_{1}^{2}(s_{1}^{2})^{2} + \ell_{2}^{2}(s_{2}^{2})^{2}}{2\ell_{1}\ell_{2}s_{1}^{2}s_{2}^{2}}\right)
    \end{align*}
    Hence, the best-case standard errors will be zero if the right-hand side in the expression above is a valid correlation, i.e. the weighted average of the square of the empirical variances as specified in the final line falls in the interval $[-1,1]$.

    Note that the expression in the last display has the same structure as:
    \begin{align*}
        \left|\frac{a^{2} + b^{2}}{2ab}\right| \leq 1 \Longleftrightarrow |a^{2} + b^{2}| \leq 2|ab|
    \end{align*}
    which is only true when $a = b$ and thus it must be that a lower bound of zero occurs if and only if $\ell_{1}s_{1}^{2} = \ell_{2}s_{2}^{2}$.   
\end{proof}

\subsection{Proof of Proposition \ref{prop_structure_location}}
\label{proof_prop_location}
We begin by proving the general statement of Proposition \ref{prop_structure_location}
\begin{proof}
    Let $f(\Omega) = \ell D \Omega D\ell'$. Optimizing over the set $C$ defined in Proposition \ref{prop_structure_location}, our problem takes the form 
    \begin{align*}
        &\min_{\Omega \in C}f(\Omega)
    \end{align*}
    We want to show that $\Omega_{\mathrm{BCS}} = \arg\min_{\Omega \in C}f(\Omega)$ is on the boundary of $C$: $\Omega_{\mathrm{BCS}} \in \mathrm{bd}(C)$ 
    
    To do this, suppose that the above statement is not true and that $\Omega^{*} = \arg\min_{\Omega \in C}f(\Omega)$ is in the interior of $C$. Thus, $\exists \epsilon > 0$ such that the ball of radius $\epsilon$ $B_{\epsilon}(\Omega^{*}) \subset C$. However, since $f$ is linear in $\Omega$, $\frac{\partial}{\partial \Omega}f(\Omega) \neq 0$ in $C$. Hence, there exists $\Omega$ that can yield a smaller value for $f$. The smallest such $\Omega$ must occur at the boundary of $C$ by the linearity of $f$. Call this matrix $\Omega_{\mathrm{BCS}}$ and we have that $\Omega_{\mathrm{BCS}} = \arg\min_{\Omega \in C}f(\Omega) \in \mathrm{bd}(C)$ i.e. the best-case correlation structure occurs at the boundary of our constraint set. 

    An identical argument proves the result for the worst-case correlation structure $\Omega_{\mathrm{WCS}}$.
\end{proof}

The $p = 2$ special case of Proposition \ref{prop_structure_location} follows from the derivation in Section~\ref{proof_prop_0bcse}. Using the expression for the standard error in the second-to-last line of this proof, we simply must find the optimal $\rho$ that will yield the largest possible and smallest possible standard errors. 
\begin{proof}
    \begin{align*}
        \frac{\partial}{\partial \rho}\ell\Sigma\ell' = \frac{\partial}{\partial \rho}\ell_{1}^{2}s_{1}^{4} + \rho \ell_{1}\ell_{2}s_{1}^{2}s_{2}^{2} + \rho \ell_{1}\ell_{2}s_{1}^{2}s_{2}^{2} + \ell_{2}^{2}s_{2}^{4} = 2\ell_{1}\ell_{2}s_{1}^{2}s_{2}^{2}
    \end{align*}
    From this first-order condition, we see that the derivative with respect to the correlation $\rho$ does not depend on $\rho$. Moreover, note that $s_{1}^{2} \geq 0$ and $s_{2}^{2} \geq 0$ as they are variances. Thus, to maximize and/or minimize this quantity linear in its choice variable, we must first determine the sign of this derivative which represents the slope of this linear function.

    The only two terms that can vary in sign are $\ell_{1}$ and $\ell_{2}$ which are multiplied together as a product. Hence, if $\mathrm{sgn}(\ell_{1}\ell_{2}) > 0$, we have an upward sloping line as a function of $\rho$, and if $\mathrm{sgn}(\ell_{1}\ell_{2}) < 0$, we have a downward sloping line. For the former case, the point on the line that yields a minimal value for the expression and retains the property of being a valid correlation is $-1$ and hence $\rho_{\mathrm{BCS}} = -1$. For the latter case, we have the opposite and so  $\rho_{\mathrm{BCS}} = 1$. Finally, as the expression is linear in $\rho$, the maximal correlation $\rho_{\mathrm{WCS}}$ simply takes the opposite signed value of  $\rho_{\mathrm{BCS}}$ for the respective case. 
\end{proof}

\end{document}